\RequirePackage{fix-cm}

\documentclass[smallextended]{svjour3} 
\usepackage{amssymb}
\usepackage{amsmath}
\usepackage{mathrsfs}
\usepackage[11pt]{extsizes}
\smartqed  
\usepackage{booktabs,caption}
\usepackage{graphicx}
\usepackage{pdflscape}
\usepackage{geometry}
 \usepackage{yfonts}
 \usepackage{mathrsfs}
\usepackage{amsfonts}
\usepackage{xcolor}
\usepackage{longtable}

 \journalname{}
\begin{document}
\title{ Skew Generalized Polycyclic Codes with Derivations}

\titlerunning{Skew Generalized Polycyclic Codes with Derivations}        

\author{Shikha Patel$^{1}$ \and Om Prakash$^{* 1}$  
}
\authorrunning{Patel and Prakash }

\institute{\at
              $^{1}$Department of Mathematics,
              Indian Institute of Technology Patna, Patna 801 106, India \\
              ORCID: 0000-0002-6512-4229 (Om Prakash)\\
              \email{shikha\_1821ma05@iitp.ac.in, om@iitp.ac.in(*corresponding author)}
                          }
\date{Received: date / Accepted: date}
\maketitle	

\begin{abstract}
	In this paper, we first consider the iterated skew polynomial ring $\mathscr{R}[z_1;\tau_1,\delta_{\tau_1}]$\\$[z_2;\tau_2,\delta_{\tau_2}]$, where $\mathscr{R}$ is a finite ring with unity. Then we use this structure for the construction of skew generalized polycyclic codes over the ring  $\mathscr{R}$ and finite field $\mathbb{F}_q$, where $q=p^m$ for some positive integer $m$. Further, we derive the structure of the generator and parity check matrices for skew generalized polycyclic codes. Furthermore, we improve the Bose-Chaudhuri-Hocquenghem (BCH) lower bound for a minimum distance of skew generalized polycyclic codes with non-zero derivations over a finite field. Moreover, we find a sufficient condition for a code to be a maximum-distance-separable (MDS) code. In addition, we provide examples of MDS codes to show the importance of our results. A comparative summary of our work with other linear codes is also discussed.
\end{abstract}

\section{Introduction}
Coding theory mainly studies the methods that are efficient and preserve the accuracy of information during transmission from one place to another. In coding theory, linear codes have been studied for the last seven decades. Initially, linear codes were studied over the binary field. In $1970$, Blake \cite{Blake} initiated the study of codes over finite rings. Cyclic codes are the most important class of linear codes from an implementation point of view and also play a prominent role in the development of algebraic coding theory. These codes have been studied over several finite rings and produce many new and optimal codes, refer \cite{Blake1,Blake,Sagar,Yildiz,Zheng}. 

In $2007$, Boucher et al. \cite{Boucher1} introduced skew cyclic codes which are a generalization of cyclic codes over a non-commutative ring, namely skew polynomial ring $\mathbb{F}_q[x;\theta]$, where $\mathbb{F}_q$ is a finite field, and $\theta$ is an automorphism of $\mathbb{F}_q$. Besides the commutative case, skew polynomial rings also have many applications in constructing algebraic codes with better parameters. It is well-known that the factorization of a polynomial plays a key role in studying these codes. More factorization leads to a rich algebraic structure that is capable of producing codes with good parameters. Therefore, it is logical to extend these codes over a structure in which polynomials have more factorization than the commutative structure. One smart choice is the skew polynomial ring. We always have
more choices to construct new codes from the codes over the skew polynomial ring than a commutative ring. That is one of the major motivations to study codes over a skew polynomial ring.

In $2009$, Boucher and Ulmer \cite{Boucher2} obtained some skew codes with Hamming distance more significant than the known linear codes with the same code rate parameters. 
This work was further advanced by Boulagouaz and Leroy in 2013 with the development of $(\sigma,\delta)$-codes \cite{MH}, and by Boucher and Ulmer \cite{Boucher} in 2014, who explored linear codes through the framework of skew polynomial rings with non-trivial automorphisms and inner derivations over finite fields. Later on, in $2017$, Cuiti$\tilde{n}$o and Tironi \cite{Luis} studied the structural properties of skew generalized cyclic codes over finite fields and also obtained some BCH type lower bounds for their minimum distances which generalizes the results of Hartmann and Tzeng \cite{Hartmann} in the noncommutative case. Moreover, cyclic, constacyclic, and negacyclic codes are part of a larger class of codes called polycyclic codes, i.e., codes that can
be viewed as ideals of a factor ring $\mathbb{F}_q[x]/(f)$, where $f$ is some polynomial of degree $n$. They reduce to cyclic
codes when $f = x ^n -1$ and to constacyclic codes when $f = x ^n -a$, for some $a \in \mathbb{F}_q^*$. Apart from this, skew polycyclic codes are a powerful generalization of all such codes. In this paper, we generalize these codes to two-dimensional skew polycyclic codes.

The two-dimensional cyclic code was introduced by Ikai et al. \cite{Ikai} in $1975$ and Imai \cite{Imai} in $1977$.  The two-dimensional theory is beneficial for the analysis and generation of two-dimensional periodic arrays. It gives a construction method for the two-dimensional feedback shift register with a minimum number of storage devices that generate a given two-dimensional periodic array. Simple 2-D linear feedback shift registers and combinatorial logic in digital VLSI circuits can be used to implement TDC codes. In \cite{Garani,Yoon}, 2-D codes for particular error patterns were examined. In $2014$, L. Xiuli et al. \cite{Li} generalized the concept as a two-dimensional skew cyclic code over a finite field. In $2019$, Sharma and Bhaintwal \cite{A} studied the structural behavior of two-dimensional skew cyclic code over the ring $\mathbb{F}_q+u\mathbb{F}_q$.  In \cite{Z}, Sepasdar and Khashyarmanesh studied the algebraic structure of two-dimensional cyclic codes of length $n=s2^k$ over the finite field. This technique was later expanded to $(\lambda_1,\lambda_2)$-constacyclic codes by Prakash and Patel \cite{Prakash}. The authors examined two-dimensional double cyclic codes over finite fields in \cite{Haji}. Other scholars have expressed interest in studying these codes (see \cite{Moro,Patel,Rajabi,Tharkal}). These works motivate us to study the structural properties of the two-dimensional skew polycyclic codes. Last but not least, BCH-type lower bounds for a minimum distance of skew generalized polycyclic codes with non-zero derivations are also considered, which generalizes the results of \cite{Hartmann,Luis}.

The remainder of the manuscript is organized as follows: Section $2$ contains some preliminaries while Section $3$ discusses skew generalized polycyclic codes and their structural properties. In Section $4$, we investigate BCH lower bounds for a minimum distance of skew generalized polycyclic codes with nonzero derivations. A sufficient condition for a code to be a maximum-distance-separable (MDS) code is also derived. Section $5$ includes some examples of MDS codes and a comparative survey of skew polycyclic codes versus polycyclic codes. Finally, Section $6$ concludes our work.

\section{Preliminaries}	
Recall that a linear code $\mathcal{L}$ of length $n$ over a ring $R$ is considered as an $R$-submodule of $R^n$ and each member of $\mathcal{L}$ is a codeword. The rank of $\mathcal{L}$ is the minimum number of generators for $\mathcal{L}$ as a module, and the free rank of $\mathcal{L}$ is the rank of $\mathcal{L}$ if $\mathcal{L}$ is free. A linear code of length $n$ and rank $k$ is denoted by $[n,k]$. The Hamming weight for a codeword $c=(c_1,c_2,\dots,c_n),$ is denoted by $w_H(c)$ and defined by $w_H(c)=\mid \{ i~: c_i\neq 0\}\mid.$ Now, the Hamming distance between two codewords $c'$ and $c''$ is defined as $d_H(c',c'')=w_H(c'-c'')$ and for $\mathcal{L}$, it is $d_H(\mathcal{L})=\min\{d_H(c',c'')~:~c'\neq c'',c',c''\in \mathcal{L}\}.$ Therefore, an $[n,k]$ linear code is compactly represented by $[n,k,d_H]$ where $d_H$ is the Hamming distance. Also, the Euclidean inner product $``\cdot "$ between any two vectors $c'=(c'_1,c'_2,\dots,c'_n)$ and $c''=(c''_1,c''_2,\dots,c''_n)$ of $R^n$ is defined by $c'\cdot c''=\sum\limits_{i=1}^nc'_ic''_i$. It is evident that the (Euclidean) dual of $\mathcal{L}$ defined by $\mathcal{L}^{\perp}=\{c'\in R^n~:~c'\cdot c''=0,$ for all $c''\in \mathcal{L}\}$ is a linear code when $\mathcal{L}$ is linear. In addition, $\mathcal{L}$ is said to be a \textit{self-orthogonal} code if $\mathcal{L}\subseteq \mathcal{L}^{\perp}$, and a \textit{self-dual} code if $\mathcal{L}=\mathcal{L}^{\perp}$.\\ If $\mathcal{L}$ is an $[n, k, d]$ code, then from
the Singleton bound, its minimum distance is bounded above by $d \leq n - k + 1.$
A code that satisfies the above bound with equality is called \emph{maximum-distance-separable (MDS)}. A code is \emph{almost MDS} if its minimum distance is one unit less than the MDS bound and \emph{optimal} if it has the highest possible minimum distance for a given length and dimension.
\begin{definition}
	 Let $\mathscr{R}$ be a finite ring with unity and $\tau$ be an automorphism of $\mathscr{R}$. Then a map $\delta_\tau : \mathscr{R} \rightarrow \mathscr{R}$ is said to be a $\tau$-derivation of $\mathscr{R}$  if
	 \begin{enumerate}
	 \item $\delta_\tau(r+s)=\delta_\tau(r)+\delta_\tau(s)$;
	 \item $\delta_\tau(rs)=\delta_\tau(r)s+\tau(r)\delta_\tau(s).$
\end{enumerate} for all $r,s\in \mathscr{R}$.
\end{definition}
 Let $\mathbb{F}_q$ be a finite field of order $q=p^m$, where $m$ is a positive integer. Suppose $\rho$ is an automorphism of $\mathbb{F}_q$ defined by $\rho(a)=a^{p^t}$ and  $\Im_\rho$ is a $\rho$-derivation of $\mathbb{F}_q$ where $t$ is a positive integer. Further, the order of automorphism $\rho$ is given by  $|\langle \rho \rangle|=\frac{m}{gcd(m,t)}$, and in particular,  $|\langle \rho \rangle|=\frac{m}{t}$ when $t|m$.  We consider the set
$$\mathbb{F}_q[z;\rho,\Im_\rho]=\{b_ez^e+\cdots+b_1z+b_0~|~ b_i\in \mathbb{F}_q~ \text{and}~ e\in \mathbb{N}\}.$$ Then $\mathbb{F}_q[z;\rho,\Im_\rho]$ is a noncommutative ring under the usual addition of polynomials and multiplication defined with respect to $za=\rho(a)z+\Im_\rho(a)$ for all  $ a\in \mathbb{F}_q.$ \\
Now, consider $$\mathscr{R}[z;\tau,\delta_\tau]=\{a_lz^l+\cdots+a_1z+a_0~|~ a_i\in \mathscr{R}~ \text{and}~ l\in \mathbb{N}\}$$ where $\tau$ is an automorphism of $\mathscr{R}$ and  $\delta_\tau$ is a $\tau$-derivation of $\mathscr{R}$. Then $\mathscr{R}[x;\tau,\delta_\tau]$ is a noncommutative ring under the usual addition of polynomials and multiplication defined with respect to $zr=\tau(r)z+\delta_\tau(r)$ for all $r\in \mathscr{R}$ and extended to all elements of $\mathscr{R}[z;\tau,\delta_\tau]$
 by associativity and distributivity. This ring is known as a skew polynomial ring.

\begin{definition} \cite{N}
	 A pseudo-linear transformation $T: \mathscr{R}^n\rightarrow \mathscr{R}^n$ is an additive map defined by
	\begin{equation}\label{4eq1}
		T= \tau(v)M + \delta_\tau(v),
	\end{equation}
	where $v= (v_1,v_2,\dots,v_n)\in \mathscr{R}^n$, $\tau(v)= (\tau(v_1),\tau(v_2),\dots,\tau(v_n))\in \mathscr{R}^n$, $M$ is a matrix of order $n\times n$ over $\mathscr{R}$ and $\delta_\tau(v)= (\delta_\tau(v_1),\delta_\tau(v_2),\dots,\delta_\tau(v_n))\in \mathscr{R}^n$. If $\delta_\tau=0$, then $T$ is known as \textit{semi-linear transformation}.
\end{definition}
Now, we consider the iterated skew polynomial ring over the ring $\mathscr{R}$. Let $S= \mathscr{R}[z_1;\tau_1,\delta_{\tau_1}] $ be a skew polynomial
ring over the ring $\mathscr{R}$, where $\tau_1$ is an automorphism and $\delta_{\tau_1}$ is a $\tau_1$-derivation of $\mathscr{R}$. If $\tau_2$ is an endomorphism and $\delta_{\tau_2}$ is a $\tau_2$-derivation of $S$, then the skew polynomial ring $B = S[z_2;\tau_2,\delta_{\tau_2}]$ is called an iterated skew polynomial ring over $\mathscr{R}$. Also,  $B = \mathscr{R}[z_1;\tau_1,\delta_{\tau_1}][z_2;\tau_2,\delta_{\tau_2}]$.

For brevity of notation, let $B_i= \mathscr{R}[z_1;\tau_1,\delta_{\tau_1}][z_2;\tau_2,\delta_{\tau_2}]\cdots[z_i;\tau_i,\delta_{\tau_i}] $ and finite sets $H=\{\tau_1,\tau_2,...,\tau_n\}$ and $D= \{\delta_{\tau_1},\delta_{\tau_2},...,\delta_{\tau_n}\}$ where $\tau_i$ is an endomorphism of $B_{i-1}$ and $\delta_{\tau_i}$ is the $\tau_i$-derivation of $B_{i-1}$, for $i=1,2,...,n$. Then, by induction, $B_i$ are the iterated skew polynomial rings over $\mathscr{R}$, $2\leq i \leq n$. If $H$ has only identity automorphism, then
$B_i$ is an iterated skew polynomial ring of derivation type, and if $D$ has only zero derivation, then $B_i$ is the iterated skew polynomial ring of automorphism type over $\mathscr{R}$.

Let $\preceq$ be the lexicographical order on $\mathbb{Z}\times \mathbb{Z}$. Now, for any $(c,d), (c',d')\in \mathbb{Z}\times \mathbb{Z}$, $(c,d)\preceq (c',d')$ if and only if $c < c'$ or $c=c'$, $d\leq d'$. Furthermore, if $(c,d)\leq (c',d')$ and $(c,d)\neq (c',d')$, then we write $(c,d)\prec (c',d')$ and $\preceq$ is a total order on $\mathbb{Z}\times \mathbb{Z}$. For any polynomial $l(z)\in \mathscr{R}[z_1;\tau_1,\delta_{\tau_1}][z_2;\tau_2,\delta_{\tau_2}]$, we define
\begin{center}
	$\mathcal{V}_l= \{(a_1,a_2)~ |~ l(z)$ contains a term $mz_1^{a_1}z_2^{a_2}, ~m\in \mathscr{R},~ m\neq 0\}$.
\end{center}
Recall from Sharma and Bhaitwal\cite{A} that the lex-degree of $l(z)$, represented by \\$lexdeg(l(z))$, is the greatest element of $\mathcal{V}_l$ with respect to the total order $\preceq$ on $\mathcal{V}_l$. The lexicographical order $\preceq$ on $\mathscr{R}[z_1;\tau_1,\delta_{\tau_1}][z_2;\tau_2,\delta_{\tau_2}]$ is defined as $l(z)\preceq g(z)$ if and only if $lexdeg(l(z))\preceq lexdeg(g(z)),$ for any $l(z),g(z)\in \mathscr{R}[z_1;\tau_1,\delta_{\tau_1}][z_2;\tau_2,\delta_{\tau_2}]$. The term of the polynomial $l(z)$ corresponding to its lex-degree is said to be the lex-leading term, and its coefficient is known as the lex-leading coefficient.

Now, we present some basic results for iterated skew polynomial rings.
\begin{lemma} \cite[Lemma~3.6]{Gr}
	Let $R$ be a ring and $\delta,$ a derivation on $R$. Let
	$S=R[z; \delta]$ be the skew polynomial ring of derivation type
	over $R$ and $\delta_1$ be another derivation of $R$. Then $\delta_1$ can be extended to a derivation
	of $S$ by $\delta_1(z)=0$ if and only if $\delta_1$ commutes with $\delta$.
\end{lemma}
\begin{theorem} \cite[Theorem~3.7]{Gr}
	Let $R$ be a ring and $D= \{\delta_1,\delta_2,...,\delta_n\},$ a finite set of derivations of $R$.
	Let $B_i$ be the set of all polynomials in indeterminates $z_{1}, z_{2}, \cdots, z_{i}$ with coefficients in $R$, for $i = 0,1,2, \cdots, n$,
	where $B_0=R$. In $B_n$, addition is usual and
	multiplication is defined with respect to $z_i r = r z_i + \delta_i(r)$ for all $r$ in $R$,
	$z_iz_j=z_jz_i$  for $i,j=1,2,...,n$.
	Then $B_i$ is a skew polynomial ring (of derivation type)
	over $B_{i-1}$, for all $i=1,2,...,n$ if and only if $\delta_i\delta_j=\delta_j\delta_i$, for
	all $i,j=1,2,...,n$.
\end{theorem}
\begin{definition}\cite{Imai}
	A binary two-dimensional code of area $s l$ is the set of $s\times l$ arrays over $\mathbb{F}_2$, called codewords or code-arrays. A two-dimensional code $\mathcal{L}$ is said to be linear if and only if $\mathcal{L}$ forms a subspace of the $sl$-dimensional space of the $s\times l$ arrays over $\mathbb{F}_2$. A two-dimensional cyclic code is defined as a two dimensional linear code such that for each code-array $C$, all
	the arrays obtained by permuting the columns or the rows of $C$ cyclically are also code-arrays.
\end{definition}
Now, we fix monic polynomials
\begin{equation}\label{6eq2}
	f(z_1,z_2)=\sum_{a_1=0}^{l}\sum_{a_2=0}^{s} r_{a_1,a_2}z_1^{a_1}z_2^{a_2}\in B=\mathscr{R}[z_1;\tau_1,\delta_{\tau_1}][z_2;\tau_2,\delta_{\tau_2}],
\end{equation}
\begin{equation}\label{6eq32}
	\textswab{f}(z_1,z_2)=\sum_{a_1=0}^{l}\sum_{a_2=0}^{s} e_{a_1,a_2}z_1^{a_1}z_2^{a_2}\in B'= \mathbb{F}_q[z_1;\rho_1,\Im_{\rho_1}][z_2;\rho_2,\Im_{\rho_2}]
\end{equation}
where $r_{a_1,a_2}\in \mathscr{R}$ and $e_{a_1,a_2}\in \mathbb{F}_q$.

Define a map $\Gamma_f : \mathscr{R}^n\rightarrow B/Bf(z_1,z_2)$ by $$\Gamma_f(c)=[c(z_1,z_2)] \in B/Bf(z_1,z_2)$$ for all $ c\in \mathscr{R}^n$ where $[c(z_1,z_2)]\in B/Bf(z_1,z_2)$ denotes the class of polynomials and $$c(z_1,z_2)=\sum_{a_1=0}^{l-1}\sum_{a_2=0}^{s-1} c_{a_1,a_2}z_1^{a_1}z_2^{a_2},$$ $c_{a_1,a_2}\in \mathscr{R}^n$, $n=ls$ and $c=(c_{ij})\in \mathscr{R}^n$  is an $l\times s$-arrays. Analogous to classical cyclic codes, we associate the array $c\in \mathscr{R}^n$ to an element $\sum_{a_1=0}^{l-1}\sum_{a_2=0}^{s-1} c_{a_1,a_2}z_1^{a_1}z_2^{a_2}$ in the quotient module $B/Bf(z_1,z_2)$. One can easily check that $\Gamma_f$ is an $\mathscr{R}$-linear isomorphism.
\begin{definition}\label{6def2}
	Let $\mathcal{L}\subseteq \mathscr{R}^n$ be a non-empty set and $n=ls$. Then
	\begin{enumerate}
		\item Suppose $\mathcal{L}$ is an linear code over $\mathscr{R}$ of length $n=sl$ in which each codeword is viewed as an $l\times s$ array.
		If $\mathcal{L}$ is closed under row $\tau_{1}$-shift and column $\tau_{2}$-shift of codewords, then  $\mathcal{L}$ is a $2$-dimensional skew cyclic code of size $l \times s$ over $\mathscr{R}$ under $\tau_{1}$ and $\tau_{2}$.
		\item  If $\mathcal{L}$ is an $\mathscr{R}$-linear code such that $T_i(\mathcal{L})\subseteq \mathcal{L}$, then $\mathcal{L}$ is an $(M_i,\tau_i,\delta_{\tau_i})$-skew polycyclic code for $i=1,2$, where $M_{1}$ is of order $l\times l$, $M_{2}$ is of order $s\times s$, and $T_i(\mathcal{L})= \{T_i(v)~|~ v\in \mathcal{L}\}$. Moreover, in the case of two-dimensional linear code, the pseudo-linear map $T_i$ is defined as follows.
		\begin{itemize}
        \item  $T_1(v_{0,a_2}, v_{1,a_2},\dots, v_{l-1,a_
				2})= \tau_1(v_{0,a_2}, v_{1,a_2},
			\dots, v_{l-1,a_
				2})M_1 + \delta_{\tau_1}(v_{0,a_2}, v_{1,a_2},\\\dots, v_{l-1,a_
				2})$, where $0 \leq a_2 \leq s-1$ and
			$$M_{1}=\begin{pmatrix}
				0 & 1  & \dots & 0\\
				\vdots & \vdots & \ddots & \vdots\\
				0 & 0  & \dots & 1\\
				-r_{0,0} & -r_{1,0} & \dots & -r_{l-1,0}
			\end{pmatrix}.$$
			\item $T_2(v_{a_{1},0}, v_{a_1,1},\dots, v_{a_1,s-1})= \tau_2(v_{a_1,0}, v_{a_1,1},\dots, v_{a_1,s-1})M_2 + \delta_{\tau_2}(v_{a_1,0}, v_{a_1,1},\\
			\dots, v_{a_1,s-1})$, where $0 \leq a_1 \leq l-1$ and
			$$M_{2}=\begin{pmatrix}
				0 & 1  & \dots & 0\\
				\vdots & \vdots & \ddots & \vdots\\
				0 & 0  & \dots & 1\\
				-r_{0,0} & -r_{0,1} & \dots & -r_{0,s-1}
			\end{pmatrix}.$$
			
		\end{itemize}
		\item $\mathcal{L}$ is an $(f,\tau_1,\tau_2,\delta_{\tau_1},\delta_{\tau_2})$-skew polycyclic code if $\mathcal{L}$ is an $\mathscr{R}$-linear code such that $T_{i_f}(\mathcal{L})\subseteq \mathcal{L}$ for $i=1,2$, where $T_{i_f}$ is as defined above, $M_{1}=(M_{1})_f$ is of order $l\times l,$ $M_{2}=(M_{2})_f$ is of order $s\times s$ defined as above and $T_{i_f}(\mathcal{L})= \{T_{i_f}(v)~|~ v\in \mathcal{L}\}$.
        \item Two $l \times s$ arrays $c=(c_{ij})\in B/Bf(z_1,z_2)$ and $d=(d_{ij})\in B/Bf(z_1,z_2)$ are said to be orthogonal if and only if
$$\sum_{i=0}^{l-1}\sum_{j=0}^{s-1}c_{ij}d_{ij}=0.$$ Let $\mathcal{L}^{\perp}$ denotes the dual of $\mathcal{L}$.  Then a polynomial $d_(z_1,z_2)\in B/Bf(z_1,z_2)$ is in $\mathcal{L}^{\perp}$ if and only if $c(z_1,z_2)d(z_1,z_2)=0\in B/Bf(z_1,z_2)$ for all codewords $c(z_1,z_2)$ of $\mathcal{L}$.
The dual code of a two-dimensional linear code $\mathcal{L}$ is the set of all $l \times s$  arrays orthogonal to all codewords of $\mathcal{L}$.
	\end{enumerate}
\end{definition}
\section{Skew generalized polycyclic codes and their properties}
In this section, we discuss the algebraic properties of  skew generalized polycyclic codes.  We generalize skew polycyclic codes  to two-dimensional skew polycyclic codes in $B = \mathscr{R}[z_1;\tau_1,\delta_{\tau_1}][z_2;\tau_2,\delta_{\tau_2}]$.
The following remarks are the direct consequence of the previous discussion.
\begin{remark}
Let $T_i$ be a pseudo-linear transformation on $\mathscr{R}^n$ for $i=1,2$. If $p(z_1,z_2)=\sum_{j=0}^{m_1}\sum_{k=0}^{m_2}p_{u,v}z_1^{j}z_2^{k}$, then $p(T_1,T_2)=\sum_{j=0}^{m_1}\sum_{k=0}^{m_2}p_{u,v}T_1^{j}T_2^k$ is not a pseudo-linear transformation.
\end{remark}
\begin{remark}\label{6ref1}
	We have $\Gamma_f(p(T_{1_f}T_{2_f})(v))=[p(z_1,z_2)v(z_1,z_2)]$ for every $p(z_1,z_2)\in B$ and $v\in \mathscr{R}^n$ such that $\Gamma_f(v)=[v(z_1,z_2)]$. This implies that the map $\Gamma_f$ is a left $B$-module isomorphism when $B/Bf(z_1,z_2)$ is a $B$-module with the product $p(z_1,z_2)[v(z_1,z_2)]=[p(z_1,z_2)v(z_1,z_2)]$ for any $p(z_1,z_2),v(z_1,z_2)\in B$.
\end{remark}
Based on the above facts, we have the following characterization of $(f,\tau_1,\tau_2,\delta_{\tau_1},\delta_{\tau_2})$-skew polycyclic codes.
\begin{theorem}
	Let $\mathcal{L}$ be a nonempty subset of $\mathscr{R}^n$. Then $\mathcal{L}$ is an $(f,\tau_1,\tau_2,\delta_{\tau_1},\delta_{\tau_2})$-skew polycyclic code if and only if $\Gamma_f(\mathcal{L})$ is a left $B$-submodule of $B/Bf(z_1,z_2)$.
\end{theorem}
\begin{proof}
	Suppose  $\mathcal{L}$ is a non-empty subset of $\mathscr{R}^n$ and an $(f,\tau_1,\tau_2,\delta_{\tau_1},\delta_{\tau_2})$-skew polycyclic code. Then $\mathcal{L}$ is $\mathscr{R}$-linear and $T_{i_f}(\mathcal{L})\subseteq \mathcal{L}$, where $T_{i_f}$ is as given in Definition (\ref{6def2}(3)) and $T_{i_f}(\mathcal{L})= \{T_{i_f}(v)~|~ v\in \mathcal{L}\}$ for $i=1,2$. For $c\in \mathcal{L}$ and $p(z_1,z_2)\in B$, we have
	\begin{align*}
		\Gamma_f(p(T_{1_f}T_{2_f})(c))=&[p(z_1,z_2)c(z_1,z_2)]~~(\text{as $\mathcal{L}$ is invariant under $T_{i_f}$})\\
		=& p(z_1,z_2)[c(z_1,z_2)]~~(\text{as~ $B/Bf(z_1,z_2)$ is a $B$-module}) \\
		=&p(z_1,z_2)	\Gamma_f(c)\in B/Bf(z_1,z_2).
	\end{align*} Therefore, $\Gamma_f(\mathcal{L})$ is a left $B$-submodule of $B/Bf(z_1,z_2)$.

	Conversely, let $\Gamma_f(\mathcal{L})$ be a left $B$-submodule of $B/Bf(z_1,z_2)$, $c\in \mathcal{L}$ and $r\in \mathscr{R}$. Since $\Gamma_f(\mathcal{L})$ is $\mathscr{R}$-linear, $\Gamma_f(rc)=r\Gamma_f(c)=r[c(z_1,z_2)]=[rc(z_1,z_2)]$. This implies that $rc\in \mathcal{L}$ and hence $\mathcal{L}$ is $\mathscr{R}$-linear. Further, from Remark \ref{6ref1}, $\Gamma_f$ is a left $B$-module isomorphism, i.e.,  $\Gamma_f(p(T_{1_f}T_{2_f})(c))=[p(z_1,z_2)c(z_1,z_2)]$ for every $p(z_1,z_2)\in B$ and $c\in \mathcal{L}$. This shows that $\mathcal{L}$ is invariant under $p(T_{1_f},T_{2_f})$ and hence $\mathcal{L}$ is invariant under $T_{i_f}$. Therefore, $\mathcal{L}$ is an $(f,\tau_1,\tau_2,\delta_{\tau_1},\delta_{\tau_2})$-skew polycyclic code.
\end{proof}
Now, we introduce the definition of skew generalized polycyclic code in $B$.
\begin{definition}
	Let $\mathcal{L}\subseteq \mathscr{R}^n$ be a nonempty set. Then $\mathcal{L}$ is said to be an $(f,\tau_1,\tau_2, \delta_{\tau_1},\delta_{\tau_2})$-skew generalized polycyclic code if $\Gamma_f(\mathcal{L})=Bg(z_1,z_2)/Bf(z_1,z_2)$, where $g(z_1,z_2)\in S[z_2;\tau_2,\delta_{\tau_2}]$ is a monic polynomial such that $f(z_1,z_2)\in Bg(z_1,z_2)$. The polynomial $g(z_1,z_2)\in B$ is said to be the generator polynomial of $\mathcal{L}$ and we write $\mathcal{L}=\langle g(z_1,z_2)\rangle_{n,q}^{\tau_1,\tau_2,\delta_{\tau_1},\delta_{\tau_2}}$.
\end{definition}
It would be useful if immediately from the generator polynomial $g(z_1,z_2)$ we could
deduce the dimension of the code and also write down a
generator matrix. The next theorem allows us to do both together.
\begin{theorem}
	Let $\mathcal{L}=\langle g(z_1,z_2)\rangle_{n,q}^{\rho_1,\rho_2,\Im_{\rho_1},\Im_{\rho_2}}$ be an $(\textswab{f},\rho_1,\rho_2,\Im_{\rho_1},\Im_{\rho_2})$-skew polycyclic code over $\mathbb{F}_q$ where  $g(z_1,z_2)=\sum_{a_1=0}^{k_1}\sum_{a_2=0}^{k_2} g_{a_1,a_2}z_1^{a_1}z_2^{a_2}$ be a monic generating polynomial of $\mathcal{L}$ in $\mathbb{F}_q[z_1;\rho_1,\Im_{\rho_1}][z_2;\rho_2,\Im_{\rho_2}]$ and $lexdeg(g(z_1,z_2))=(k_1,k_2)$. Then
	$$	\mathcal{B'}=\begin{Bmatrix}
		g(z_1,z_2),& z_1g(z_1,z_2),&\dots&, z_1^{l-k_1-1}g(z_1,z_2)\\
		z_2g(z_1,z_2),& z_2z_1g(z_1,z_2),& \dots&, z_2z_1^{l-k_1-1}g(z_1,z_2)	\\
		\vdots& \vdots &\dots& \vdots\\
		z_2^{s-k_2-1}g(z_1,z_2),& z_2^{s-k_2-1}z_1g(z_1,z_2),& \dots&, z_2^{s-k_2-1}z_1^{l-k_1-1}g(z_1,z_2)
	\end{Bmatrix}$$
	is a basis for $\mathcal{L}$ over $\mathbb{F}_q$ and $|\mathcal{L}|=q^{(l-k_1)(s-k_2)}$.
\end{theorem}
\begin{proof}
	Let $\mathcal{L}=\langle g(z_1,z_2)\rangle_{n,q}^{\rho_1,\rho_2,\Im_{\rho_1},\Im_{\rho_2}}$ be an $(\textswab{f},\rho_1,\rho_2,\Im_{\rho_1},\Im_{\rho_2})$-skew polycyclic code over $\mathbb{F}_q$ and $0\neq p(z_1,z_2)\in \mathcal{L}=\langle g(z_1,z_2)\rangle_{n,q}^{\rho_1,\rho_2,\Im_{\rho_1},\Im_{\rho_2}}$. Then there exists $q(z_1,z_2)$ in $B'$ such that $p(z_1,z_2)=q(z_1,z_2)g(z_1,z_2)$. Also, $lexdeg(p(z_1,z_2))\leq (l-1,s-1)$. Hence,
	\begin{align*}
		lexdeg(q(z_1,z_2))=&lexdeg(p(z_1,z_2))-lexdeg(g(z_1,z_2))\\
		&\leq (l-1,s-1)-(k_1,k_2)\\=&(l-k_1-1,s-k_2-1).
	\end{align*}
	Hence, $\mathcal{B'}$ generates $\mathcal{L}$. Next, suppose $\{r_{i,j}(z_1,z_2)|0\leq i\leq l-k_1-1,~ 0\leq j\leq s-k_2-1\}$ is a subset of $\mathbb{F}_q[z_1;\rho_1,\Im_{\rho_1}][z_2;\rho_2,\Im_{\rho_2}]$ such that $$\sum_{i=0}^{l-k_1-1}\sum_{j=0}^{s-k_2-1} e_{i,j}z_1^{i}z_2^{j}g(z_1,z_2)=0.$$ If possible, let $r=\sum_{i=0}^{l-k_1-1}\sum_{j=0}^{s-k_2-1} e_{i,j}$ such that $r\neq 0$. Then $lexdeg~ r\leq (l-k_1-1,s-k_2-1)$ and $rg(z_1,z_2)=0$. Therefore, there exists $t(z_1,z_2) \in \mathbb{F}_q[z_1;\rho_1,\Im_{\rho_1}][z_2;\rho_2,\Im_{\rho_2}]$ such that $rg(z_1,z_2)=0 = t(z_1,z_2)f(z_1,z_2)$ where $\textswab{f}=\sum_{a_1=0}^{l}\sum_{a_2=0}^{s} e_{a_1,a_2}z_1^{a_1}z_2^{a_2}$ and $e_{a_1,a_2}\in \mathbb{F}_q$. Hence,
	\begin{align*}
		lexdeg(t(z_1,z_2)f(z_1,z_2))=&lexdeg(r(z_1,z_2))+lexdeg(g(z_1,z_2))\\
		=&(l-k_1-1,s-k_2-1)+(k_1,k_2)\\
		=& (l-1,s-1),
	\end{align*}
	which is a contradiction. Hence, $r=\sum_{i=0}^{l-k_1-1}\sum_{j=0}^{s-k_2-1} e_{i,j}=0$, implies that $e_{i,j}=0$ for all $0\leq i\leq l-k_1-1,~ 0\leq j\leq s-k_2-1$. Therefore,  $\mathcal{B'}$ is a basis for $\mathcal{L}$ over $\mathbb{F}_q$. Thus, $|\mathcal{L}|=q^{(l-k_1)(s-k_2)}$.
\end{proof}
Let $\mathcal{L}=\langle g(z_1,z_2)\rangle_{n,q}^{\rho_1,\rho_2,\Im_{\rho_1},\Im_{\rho_2}}$ be an $(\textswab{f},\rho_1,\rho_2,\Im_{\rho_1},\Im_{\rho_2})$-skew generalized polycyclic code over $\mathbb{F}_q$. From above theorem, we can conclude that $\mathcal{L}$ is a free left $\mathbb{F}_q$-module of dimension $k=(l-k_1)(s-k_2)=n-lk_2-sk_1+k_1k_2$.  Now, we define
\begin{equation}\label{eq3}
	\mathbb{G}_1=
	\begin{pmatrix}
		g\\
		T_{1_\textswab{f}}(g)\\
		\vdots\\
		T_{1_\textswab{f}}^{l-k_1-1}(g)
	\end{pmatrix}
\end{equation}
where $g(z_1,z_2)=\sum_{a_1=0}^{k_1}\sum_{a_2=0}^{k_2} g_{a_1,a_2}z_1^{a_1}z_2^{a_2}\in B'$ and $i=1,2$, i.e., $g$ is an $(k_1+1)\times (k_2+1)$ arrays defined as
$$g= \begin{pmatrix}
	g_{0,0} & g_{0,1}  & \dots & g_{0,k_2}\\
	g_{1,0} & g_{1,1}  & \dots & g_{1,k_2}\\
	\vdots & \vdots & \ddots &  \vdots \\
	g_{k_1,0} & g_{k_1,1} & \dots & g_{k_1,k_2}
\end{pmatrix}.$$ If $[g(z_1,z_2)]=\Gamma_f(g),$ then generator matrix of $\mathcal{L}$ is given by
\begin{equation}\label{eq3}
	\mathbb{G}=
	\begin{pmatrix}
		\mathbb{G}_1\\
		T_{2_\textswab{f}}(	\mathbb{G}_1)\\
		\vdots\\
		T_{2_\textswab{f}}^{s-k_2-1}(	\mathbb{G}_1)
	\end{pmatrix}.
\end{equation}
It is well-known that $dim(\mathcal{L})+dim(\mathcal{L}^\perp)=sl$. Therefore, $dim(\mathcal{L}^\perp)=k'=lk_2+sk_1-k_1k_2$. Now, we have the following theorem for the dual code $\mathcal{L}^\perp$ of $(\textswab{f},\rho_1,\rho_2,\Im_{\rho_1},\Im_{\rho_2})$-skew polycyclic code $\mathcal{L}\subseteq \mathbb{F}_q^n$.
\begin{theorem}
	Let $\textswab{f}(z_1,z_2)\in B'$, as in Equation (\ref{6eq2}), $T_{i_\textswab{f}}$ be its associated pseudo-linear transformation given in Definition (\ref{6def2}(3)) for $i=1,2$ and $\mathscr{R}=\mathbb{F}_q$. If $\textswab{f}(z_1,z_2)=h(z_1,z_2)g(z_1,z_2)\\=g(z_1,z_2)h'(z_1,z_2)$ for some monic skew polynomials $g(z_1,z_2),h(z_1,z_2),h'(z_1,z_2)\in B'$ and $e_{0,0}\neq 0$, then linearly independent columns of the matrix
	$$\mathbb{H}=
	\begin{pmatrix}
		\mathbb{H}_1\\
		T_{1_\textswab{f}}(\mathbb{H}_1)\\
		\vdots\\
		T_{1_\textswab{f}}^{l-1}(\mathbb{H}_1)
	\end{pmatrix}$$
	form a basis of $\mathcal{L}^\perp$ where
	$$\mathbb{H}_1=
	\begin{pmatrix}
		h'\\
		T_{2_\textswab{f}}(h')\\
		\vdots\\
		T_{2_\textswab{f}}^{s-1}(h')
	\end{pmatrix},$$
	$\Gamma_\textswab{f}(h')=[h'(z_1,z_2)]$, $lexdeg(h'(z_1,z_2))=(a,b)$ and $\mathcal{L}=\langle g(z_1,z_2)\rangle_{n,q}^{\rho_1,\rho_2,\Im_{\rho_1},\Im_{\rho_2}}$.
\end{theorem}
\begin{proof}
	Let $a\in \mathcal{L}=\langle g(z_1,z_2)\rangle_{n,q}^{\rho_1,\rho_2,\Im_{\rho_1},\Im_{\rho_2}}$
	and $\Gamma_\textswab{f}(h')=[h'(z_1,z_2)]$. So, $\Gamma_\textswab{f}(a)=[a(z_1,z_2)]=[b(z_1,z_2) g(z_1,z_2)]$, for some $b(z_1,z_2), a(z_1,z_2)\in B'$. For $i=1,2$,  we have $\Gamma_\textswab{f}(a(T_{1_\textswab{f}},T_{2_\textswab{f}})(h'))=[a(z_1,z_2)h'(z_1,z_2)]=[(b(z_1,z_2) g(z_1,z_2))h'(z_1,z_2)]=[b(z_1,z_2) (g(z_1,z_2)h'(z_1,z_2))]\\=[b(z_1,z_2) \textswab{f}(z_1,z_2)]=[0] \in B'/B'f(z_1,z_2)$. This implies $a(T_{1_\textswab{f}},T_{2_\textswab{f}})(h')=0$. Therefore, $$0=a(T_{1_\textswab{f}},T_{2_\textswab{f}})(h')= \sum_{a_1=0}^{l-1}\sum_{a_2=0}^{s-1} a_{a_1,a_2}T_{1_\textswab{f}}^{a_1}T_{2_\textswab{f}}^{a_2}(\bar{h'})$$ where $a(z_1,z_2)=\sum_{a_1=0}^{l-1}\sum_{a_2=0}^{s-1} a_{a_1,a_2}z_1^{a_1}z_2^{a_2}$. This shows that $a\mathbb{H}^t=0$ for any $a\in \mathcal{L}$ and $\mathbb{H}^t$  is the transpose matrix of $\mathbb{H}$. Further, we note that
	\begin{equation*}
		\Gamma_\textswab{f}((T_{1_\textswab{f}}^{e}T_{2_\textswab{f}}^{j})(h'))=z_1^{e}z_2^{j}h'(z_1,z_2)	=
		\begin{cases}
			z_1^{e}z_2^{j}h'(z_1,z_2) & \text{for $0\leq e\leq k_2-1, 0\leq j\leq l-1$}\\
			z_1^{e}z_2^{j}h'(z_1,z_2) & \text{for $k_2\leq e\leq s-1, 0\leq j\leq k_1-1$}
		\end{cases}
	\end{equation*}
	where $lexdeg(h'(z_1,z_2))=(a,b)$.  Hence,
	$$	\begin{Bmatrix}
		h',&T_{2_\textswab{f}}(h'),&\dots&,
		T_{2_\textswab{f}}^{lk_2-1}(h'),\\
		T_{2_\textswab{f}}(h'),&
		T_{2_\textswab{f}}^2(h'),&
		\dots&,
		T_{2_\textswab{f}}^{s-k_2}(h'),\\
		T_{1_\textswab{f}}	T_{2_\textswab{f}}(h'),&
		T_{1_\textswab{f}}T_{2_\textswab{f}}^2(h'),&
		\dots&,
		T_{1_\textswab{f}}T_{2_\textswab{f}}^{s-k_2}(h'), \\
		\vdots& \vdots &\vdots\\
		T_{1_\textswab{f}}^{k_1-1}T_{2_\textswab{f}}(h'),&
		T_{1_\textswab{f}}^{k_1-1}T_{2_\textswab{f}}^2(h'),&
		\dots&,
		T_{1_\textswab{f}}^{k_1-1}T_{2_\textswab{f}}^{s-k_2}(h')
	\end{Bmatrix}$$	are independent and form a basis for $\mathcal{L}^\perp$.
\end{proof}
Also, when $\mathscr{R}$ is commutative, the following theorem gives the idea of the dual code $\mathcal{L}^\perp$ of $(M_i,\tau_i,\delta_{\tau_i})$-skew polycyclic code $\mathcal{L}\subseteq \mathscr{R}^n$ for $i=1,2$.
\begin{theorem}
	Suppose $\mathscr{R}$ is commutative and $\delta_{\tau_i}: \mathscr{R}\rightarrow \mathscr{R}$ is a $\tau_i$-derivation, where $\tau_i\in Aut(\mathscr{R})$ for $i=1,2$. If $\mathcal{L}\subseteq \mathscr{R}^n$ is an $(M_i,\tau_i,\delta_{\tau_i})$-skew polycyclic code, then $\mathcal{L}^\perp$ is a $((M_i^t)_{\tau_i^{-1}}, \tau_i^{-1}, \delta_{\tau_i}'')$-skew polycyclic code with $\tau_i^{-1}$-derivation $\delta_{\tau_i}''=-\tau_i^{-1}\delta_{\tau_i}$ for $i=1,2$, where for a matrix $E=[(a_{mn})],$ we denote $E_{\tau_i^{-1}}=[\tau_i^{-1}(a_{mn})]$ and $E^t$ its transpose matrix.
\end{theorem}
\begin{proof}
	We define two pseudo-linear transformations $T_i$ and $T_i'$ by $T_i(v)=\tau_i(v)\cdot~ M_i+\delta_{\tau_i}(v)$ and $T_i'(v)=\tau_i^{-1}(v)\cdot~(M_i^t)_{\tau_i^{-1}}+\delta_{\tau_i}''(v)$ for $i=1,2$ for every $v\in \mathscr{R}^n$, respectively. For any $a\in \mathcal{L}$ and $b\in \mathcal{L}^\perp$, we have $0=b\cdot~ T_i^t(a)=b\cdot~ (\tau_i(a)\cdot~ M_i)^t +b\cdot~  (\delta_{\tau_i}(a))^t=b\cdot~ (\tau_i(a)\cdot~ M_i)^t-\delta_{\tau_i}(b)\cdot~ (\tau_i(a))^t.$ This implies $0=b\cdot~ M_i^t\cdot~ (\tau_i(a))^t-\delta_{\tau_i}(b)\cdot~ (\tau(a))^t= (b\cdot~ M_i^t-\delta_{\tau_i}(b))\cdot~ (\tau_i(a))^t$, i.e., $(b\cdot~ M_i^t-\delta_{\tau_i}(b))\cdot~ (\tau_i(a))^t=0$, for $i=1,2$. Therefore, $0=\tau_i^{-1}(0)=\tau_i^{-1}(b\cdot~ M_i^t-\delta_{\tau_i}(b))\cdot~ (a)^t=(\tau_i^{-1}(b)\cdot~ (M_i^t)_{\tau_i^{-1}}-\tau_i^{-1}\delta_{\tau_i}(b))\cdot~(a)^t=(\tau_i^{-1}(b)\cdot~ (M_i^t)_{\tau_i^{-1}}+\delta_{\tau_i}''(b))\cdot~(a)^t=T_i'(b)\cdot~(a)^t$, i.e., $T_i'(b)\cdot~(a)^t=0$ for all $a\in \mathcal{L}$, $b\in \mathcal{L}^\perp$ and $i=1,2$.
\end{proof}
The following examples provide the novelty of our results and the reason we are studying two-dimensional linear codes, especially $(\textswab{f},\rho_1,\rho_2,\Im_{\rho_1}, \Im_{\rho_2}
)$-skew polycyclic codes. Here, all computations are carried out by using Magma software \cite{Bosma}.
\begin{example}
	Let  $\mathbb{F}_4=\mathbb{F}_2(w)$ where $w^2+w+1=0$ and $B=\mathbb{F}_q[z_1;\rho_1,\Im_{\rho_1}][z_2;\rho_2,\Im_{\rho_2}]$ where $\rho_1,~\rho_2: \mathbb{F}_4 \rightarrow \mathbb{F}_4$ are Frobenius automorphisms defined by $\rho_1(a)=\rho_2(a)=a^2$,  $\Im_{\rho_1}: \mathbb{F}_4 \rightarrow \mathbb{F}_4$ is a $\rho_1$-derivation defined by $\Im_{\rho_1}(a)=\rho_1(a)-a$ and $~\Im_{\rho_2}$ is a zero derivation. Let $$f=w^2z_2^3+z_1^2+z_1^2z_2^2+w^2z_1^2z_2+w^2z_2+1=(w^2z_2+z_1^2+1)(w^2z_2+z_2^2+1)$$ and $g=z_2^2+w^2z_2+1$. We have $l=2$, $s=3$, $k_1=0$, $k_2=2$
	$$g=\begin{pmatrix}
	1&w^2&1\\
	0&0&0
	\end{pmatrix},~ M_1=\begin{pmatrix}
	0&1\\
	-1&0
	\end{pmatrix},~
	M_2=\begin{pmatrix}
	0&1&0\\
	0& 0&1\\
	-1&-w^2&0
	\end{pmatrix}.$$ Also, $n=ls=6$ and $k=(l-k_1)(s-k_2)=(2-0)(3-2)=2$. From, Equation \ref{eq3} generator matrix of $(\textswab{f},\rho_1,\rho_2,\Im_{\rho_1}, \Im_{\rho_2}
)$-skew generalized polycyclic code $\mathcal{L}$ is given by
	$$\mathbb{G}=\begin{pmatrix}
	g\\
	T_{1_\textswab{f}}(g)\\
	\end{pmatrix}=\begin{pmatrix}
	1&0&w^2&0&1&0\\
	0&1&w^2&1&0&1\\
	\end{pmatrix}.$$ Hence, $\mathcal{L}$ is a $(\textswab{f},\rho_1,\rho_2,\Im_{\rho_1}, \Im_{\rho_2}
)$-skew generalized polycyclic code having parameter $[6,2,3]$ which is near to optimal \cite{Grassl}.
\end{example}

\begin{example}
	Let $\mathbb{F}_9=\mathbb{F}_3(w)$, where $w^2+2w+2=0$ and  $B=\mathbb{F}_q[z_1;\rho_1,\Im_{\rho_1}][z_2;\rho_2,\Im_{\rho_2}]$ where $\rho_1,~\rho_2: \mathbb{F}_9 \rightarrow \mathbb{F}_9$  are Frobenius automorphisms defined by $\rho_1(a)=\rho_2(a)=a^3$, $\Im_{\rho_1}$ and $~\Im_{\rho_2}$ are zero derivations. Let $$\textswab{f}=w^3z_2^3+w^6z_2+w^3z_2^2+w^3+w^6z_1^3z_2^2+w^3z_1^3z_2+z_1^3=(w^5z_2+z_1^3+w^3)(w^2z_2^2+wz_2+1)$$ and $g=w^2z_2^2+wz_2+1$. We have $l=3$, $s=3$, $k_1=0$, $k_2=2$
	$$g=\begin{pmatrix}
	1&w&w^2\\
	0&0&0\\
	0&0&0
	\end{pmatrix},~ M_1=\begin{pmatrix}
	0&1&0\\
	0&0&1\\
	-w^3&0&0
	\end{pmatrix},~
	M_2=\begin{pmatrix}
	0&1&0\\
	0&0&1\\
	-w^3&-w^6&-w^3
	\end{pmatrix}.$$ Also, $n=ls=9$ and $k=(l-k_1)(s-k_2)=(3-0)(3-2)=3$. From, Equation \ref{eq3} generator matrix of $(\textswab{f},\rho_1,\rho_2,\Im_{\rho_1}, \Im_{\rho_2}
)$-skew generalized polycyclic code $\mathcal{L}$ is given by
	$$\mathbb{G}=\begin{pmatrix}
	g\\
	T_{1_\textswab{f}}(g)\\
	T_{1_\textswab{f}}^2(g)\\
	\end{pmatrix}=\begin{pmatrix}
	1&0&0&w&0&0&w^2&0&0\\
	0&1&0&0&w^3&0&0&w^6&0\\
	0&0&1&0&0&w&0&0&w^2\\
	\end{pmatrix}.$$ Hence, $\mathcal{L}$ is a $(\textswab{f},\rho_1,\rho_2,\Im_{\rho_1}, \Im_{\rho_2}
)$-skew generalized polycyclic code having parameter $[9,3,6]$ which is almost MDS.
\end{example}

\section{BCH lower bounds for the minimum distance of skew polycyclic codes with non-zero derivations}
In this section, we investigate BCH lower bounds for the minimum distance of skew polycyclic codes with non-zero derivations. To do so, first we derive the series of lemmas useful towards the main result of this section.
Following the norm used by Boucher and Ulmer \cite{Boucher2},
the norm of $i\in \mathbb{N},~N_{i}^{\rho, \Im_\rho}(b)$ is recursively defined as
$$N_{0}^{\rho, \Im_\rho}(b)=1$$ and
$$N_{i+1}^{\rho, \Im_\rho}(b)= \rho(N_{i}^{\rho, \Im_\rho}(b))b+ \Im_\rho (N_{i}^{\rho, \Im_\rho}(b)).$$
In particular, if $\Im_\rho=0$, we get the classical norm $N_{i}^{\rho, \Im_\rho}(b)= b\rho(b)\dots \rho^{i-1}(b).$

In this section, we derive the BCH lower bounds for the minimum distance of the skew polycyclic codes of the ring  $\mathcal{A}=\mathbb{F}_q [z; \rho,\Im_\rho]$ under certain conditions, where $\Im_\rho$ is a $\rho$-derivation.

Throughout this section, we consider $\textswab{f}(z)=z^n-\textswab{f}_{n -1}z^{n -1}-\cdots-\textswab{f}_1z-\textswab{f}_{0}$, where $\textswab{f}_{n -1},\dots,\textswab{f}_1,\textswab{f}_{0}\in \mathbb{F}_q$ and $f_{0}\neq 0$. Again, consider a map $\Omega_\textswab{f} : \mathbb{F}_q^n\rightarrow \mathcal{A}/\mathcal{A}\textswab{f}(z)$ by $\Omega_\textswab{f}(c)=[c(z)] \in \mathcal{A}/\mathcal{A}\textswab{f}(z)$ for all $ c\in \mathbb{F}_q^n$ where $$c(z)=\sum_{j=0}^{n -1} c_{j}z^{j},$$ $c_{j}\in \mathbb{F}_q$, i.e., $\Omega_\textswab{f}(c)=\Omega_\textswab{f}((c_0,c_1,\dots, c_{n -1}))=[c(z)] \in \mathcal{A}/\mathcal{A}\textswab{f}(z)$ $\forall$ $(c_0,c_1,\dots, c_{n -1})\in \mathbb{F}_q^n$. The map $\Omega_\textswab{f}$ is an $\mathbb{F}_q$-linear isomorphism between $\mathbb{F}_q$-modules. Now, we provide the following two lemmas that will be used later in the proof of the main result.
\begin{lemma}\label{6lem1}
	For any  non-zero $\rho$-derivation $\Im_\rho$, in $\mathcal{A}/\mathcal{A}\textswab{f}(z)$ we have $P(z)z=1$ and $zQ(z)=1+ \Im_\rho(Q(z))$,
	where $P(z)= \textswab{f}_{0}^{-1}z^{n -1}-\textswab{f}_{0}^{-1}\textswab{f}_{n -1}z^{n-2}-\cdots-\textswab{f}_{0}^{-1}\textswab{f}_2z-\textswab{f}_{0}^{-1}\textswab{f}_1$ and $Q(z)=\rho^{-1}(\textswab{f}_{0}^{-1})z^{n -1}-\rho^{-1}(\textswab{f}_{0}^{-1}\textswab{f}_{n -1})z^{n-2}-\cdots-\rho^{-1}(\textswab{f}_{0}^{-1}\textswab{f}_2)z-\rho^{-1}(\textswab{f}_{0}^{-1}\textswab{f}_1)$.
	
\end{lemma}
\begin{proof}
	Note that in $\mathcal{A}/\mathcal{A}\textswab{f}(z),$ we have the following equivalences:
	$$z^n-\textswab{f}_{n -1}z^{n -1}-\cdots-\textswab{f}_1z-\textswab{f}_{0}=0,$$ i.e., $$z^n-\textswab{f}_{n -1}z^{n -1}-\cdots-\textswab{f}_1z=\textswab{f}_{0}.$$ As $\textswab{f}_{0}\neq 0$,
	$$(\textswab{f}_0)^{-1}(z^{n -1}-\textswab{f}_{n -1}z^{n-2}-\cdots-\textswab{f}_1)z = 1,$$ i.e.,
	~$$(\textswab{f}_{0}^{-1}z^{n -1}-\textswab{f}_{0}^{-1}\textswab{f}_{n -1}z^{n-2}-\cdots-\textswab{f}_{0}^{-1}\textswab{f}_2z-\textswab{f}_{0}^{-1}\textswab{f}_1)z= 1.$$
	Hence, $P(z)z =1$.

	Moreover,
	$$z^n-\textswab{f}_{n -1}z^{n -1}-\cdots-\textswab{f}_1z =\textswab{f}_{0},$$ i.e., $$\textswab{f}_{0}^{-1}z^{n}-\textswab{f}_{0}^{-1}\textswab{f}_{n -1}z^{n -1}-\cdots-\textswab{f}_{0}^{-1}\textswab{f}_2z^2-\textswab{f}_{0}^{-1}\textswab{f}_1z= 1.$$ This implies
	$\textswab{f}_{0}^{-1}zz^{n -1}-\textswab{f}_{0}^{-1}\textswab{f}_{n -1}zz^{n-2}-\cdots-\textswab{f}_{0}^{-1}\textswab{f}_1z= 1$, i.e., $z\rho^{-1}(\textswab{f}_{0}^{-1})z^{n -1}-z\rho^{-1}(\textswab{f}_{0}^{-1}\textswab{f}_{n -1})z^{n-2}-\cdots-z\rho^{ -1}(\textswab{f}_{0}^{-1}\textswab{f}_1)= 1 + \Im_\rho(\rho^{-1}(\textswab{f}_{0}^{-1}))z^{n -1} - \Im_\rho(\rho^{-1}(\textswab{f}_{0}^{-1}\textswab{f}_{n -1}))\\z^{n-2} -\cdots - \Im_\rho(\rho^{-1}(\textswab{f}_{0}^{-1}\textswab{f}_1))$. This gives
	$z(\rho^{-1}(\textswab{f}_{0}^{-1})z^{n -1}-\rho^{-1}(\textswab{f}_{0}^{-1}\textswab{f}_{n -1})z^{n-2}-\cdots-\rho^{ -1}(\textswab{f}_{0}^{-1}\textswab{f}_1))= 1 + \Im_\rho(\rho^{-1}(\textswab{f}_{0}^{-1}))z^{n -1} - \Im_\rho(\rho^{-1}(\textswab{f}_{0}^{-1}\textswab{f}_{n -1}))z^{n-2} -\cdots - \Im_\rho(\rho^{-1}(\textswab{f}_{0}^{-1}\textswab{f}_1))$.\\ Hence, $zQ(z) = 1+\Im_\rho(Q(z))$
	where $Q(z)=\rho^{-1}(\textswab{f}_{0}^{-1})z^{n -1}-\rho^{-1}(\textswab{f}_{0}^{-1}\textswab{f}_{n -1})z^{n-2}-\cdots-\rho^{-1}(\textswab{f}_{0}^{-1}\textswab{f}_2)z-\rho^{-1}(\textswab{f}_{0}^{-1}\textswab{f}_1)$.

	Further, if $\rho = id.$, then $Q(z)=\textswab{f}_{0}^{-1}z^{n -1}-\textswab{f}_{0}^{-1}\textswab{f}_{n -1}z^{n-2}-\cdots-\textswab{f}_{0}^{-1}\textswab{f}_2x-\textswab{f}_{0}^{-1}\textswab{f}_1 =P(z)$.
\end{proof}
\begin{lemma}\label{6lem2}
	Suppose $\mathcal{L}\subseteq \mathbb{F}_q^n$ is a skew polycyclic code, and consider a polynomial $c(z)\in \Omega_\textswab{f}(\mathcal{L})$ with $w_H(c) = w$. Then there exists $r(z)\in \mathcal{A}$ such that \\
	$$t(z)=r(z)c(z) = 1+ \sum_{i = 1}^{w-1} c_{i}z^{a_{i}} + \sum_{i = 1}^{w-1} d_{i}z^{k_{i}} \in\Omega_\textswab{f}(\mathcal{L}),$$
	where $c_{i}, ~ d_i \in \mathbb{F}_q^{*}$ and $a_{i}, ~ k_i\in \mathbb{N}$ with $a_{i}, ~ k_i \leq n -1$ for $i = 1,\dots,w-1$. Also, $$w_H(t )=w.$$
\end{lemma}
\begin{proof}
	Since $w_H(c) = w$, we can write $c(z)= b_{k_0}z^{k_0} +  b_{k_1}z^{k_1} +\cdots+ b_{k_{w-1}}z^{k_{w-1}}$,
	where $b_{k_j}\in \mathbb{F}_q $, $k_{0}<\dots< k_{w-1}$ and $b_{k_j}\neq 0$, for $j= 0,\dots,w-1$. Hence,
	\begin{align*}
		c(z)=&  b_{k_0}z^{k_0}(1+ z^{-k_0} b_{k_0}^{-1}b_{k_1}z^{k_1} +\cdots+z^{-k_0} b_{k_0}^{-1}b_{k_{w-1}}z^{k_{w-1}})\\
		& = b_{k_0}z^{k_0}(1+\rho^{-k_{0}}(b_{k_0}^{-1}b_{k_1})z^{k_1-k_0}-\Im_\rho(\rho^{-k_0}(b_{k_0}^{-1}b_{k_1}))z^{k_1}+\cdots \\&+
		\rho^{-k_{0}}(b_{k_0}^{-1}b_{k_{w-1}})
		z^{k_{w-1}-k_0} -\Im_\rho(\rho^{-k_0}(b_{k_0}^{-1}b_{k_{w-1}}))z^{k_{w-1}})\\
		& = b_{k_0}z^{k_0}(1+\rho^{-k_{0}}(b_{k_0}^{-1}b_{k_1})z^{k_1-k_0}+\cdots+\rho^{-k_{0}}(b_{k_0}^{-1}b_{k_{w-1}})z^{k_{w-1}-k_0}\\&-\Im_\rho(\rho^{-k_0}(b_{k_0}^{-1}b_{k_1}))z^{k_1}-\cdots -\Im_\rho(\rho^{-k_0}(b_{k_0}^{-1}b_{k_{w-1}}))z^{k_{w-1}})\\
		z^{-k_0} b_{k_0}^{-1}c(z) =&1+\rho^{-k_{0}}(b_{k_0}^{-1}b_{k_1})z^{k_1-k_0}+\cdots+\rho^{-k_{0}}(b_{k_0}^{-1}b_{k_{w-1}})z^{k_{w-1}-k_0}\\&-\Im_\rho(\rho^{-k_0}(b_{k_0}^{-1}b_{k_1}))z^{k_1}-\cdots -\Im_\rho(\rho^{-k_0}(b_{k_0}^{-1}b_{k_{w-1}}))z^{k_{w-1}}.
	\end{align*}
	In Lemma \ref{6lem1},
	$P(z)z =1$. This implies $z^{-k_0}= P^{k_0}(z)$.

	Therefore, $ P^{k_0}(z) b_{k_0}^{-1}c(z) =1+\rho^{-k_{0}}(b_{k_0}^{-1}b_{k_1})z^{k_1-k_0}+\cdots+\rho^{-k_{0}}(b_{k_0}^{-1}b_{k_{w-1}}z^{k_{w-1}-k_0}) \\-\Im_\rho(\rho^{-k_0}(b_{k_0}^{-1}b_{k_1}))z^{k_1}-\cdots -\Im_\rho(\rho^{-k_0}(b_{k_0}^{-1}b_{k_{w-1}}))z^{k_{w-1}}\in \Omega_\textswab{f}(\mathcal{L})\subseteq \mathcal{A}/\mathcal{A}\textswab{f}(z)$.\\
	Now, put $r(z)=P^{k_0} b_{k_0}^{-1}$, $c_{j}=\rho^{-k_{0}}(b_{k_0}^{-1}b_{k_{j}})$, $d_{j}=- \Im_\rho(\rho^{-k_{0}}(b_{k_0}^{-1}b_{k_{j}}))=-\Im_\rho(c_j)$ and $a_{j}=k_j-k_0 $. This gives
	$$t(z)=r(z)c(z) = 1+ \sum_{i = 1}^{w-1} c_{i}z^{a_{i}} + \sum_{i = 1}^{w-1} d_{i}z^{k_{i}} \in\Omega_\textswab{f}(\mathcal{L}).$$ From above, we have $w_H(t )= w.$
\end{proof}
Inspired by the work done by Hartmann and Tzeng \cite{Hartmann}, Shi et al. \cite{Shi},  Cuiti$\tilde{n}$o and Tironi \cite{Luis}, the following result provide the lower bound on the minimum Hamming distance of skew polycyclic  codes in $\mathbb{F}_q[z;\rho,\Im_\rho]$.
\begin{theorem}\label{6t1}
	Let $\mathcal{L}=\langle g(z)\rangle_{n,q}^{\rho,\Im_\rho}$ be an skew polycyclic code. Suppose there exists $ \beta\in \bar{ \mathbb{F}_q}$( the algebraic closure of $\mathbb{F}_q$) and $l,m \in\mathbb{Z}_{+}\cup \{0\} $ such that $g(\beta_{i}^{l+mi})=0$, where $\beta_{i}= N_{l+mi}^{\rho,\Im_\rho}(\beta)$ for $i= 0,\dots,\bigtriangleup-2$. If $N_{i}(\beta^{m})\neq 1$, $N_{l+mi-1}(N_{i}(\beta^{m}))=1$ and $N_{i}^{\rho,\Im_\rho}(N_{l+mi}^{\rho,\Im_\rho}(\beta))=N_{l+mi}^{\rho,\Im_\rho}(N_{i}^{\rho,\Im_\rho}(\beta))$ for all $i= 1,\dots,n -1$, then $d_{\mathcal{L}}\geq \bigtriangleup$.
\end{theorem}
\begin{proof}
	If possible, let there exists  a polynomial $c(z)\in\Omega_\textswab{f}(\mathcal{L})$ with $w_H(c)=w< \bigtriangleup$. Then from Lemma \ref{6lem2}, $c(z)$ can be represented as
	\begin{align*}
		c(z)=& 1 +  \sum_{i = 1}^{w-1} c'_{i}z^{a_i} + \sum_{i = 1}^{w-1} d_{i}z^{k_i}
		= 1+  \sum_{i = 1}^{w-1} c'_{i}z^{a_{i}} -\sum_{i = 1}^{w-1} \Im_\rho( c'_{i})z^{k_{i}}\\
		=& 1+  \sum_{i = 1}^{w-1} c'_{i}z^{k_{i}-k_0} -\sum_{i = 1}^{w-1} \Im_\rho( c'_{i})z^{k_{i}}
		= 1+  \sum_{i = 1}^{w-1} (c'_{i}z^{-k_0} - \Im_\rho( c'_{i}))z^{k_i}\\
		=& 1+  \sum_{i = 1}^{w-1} (c'_{i}P(z)^{k_0} - \Im_\rho( c'_{i}))z^{k_i}
		= 1+  \sum_{i = 1}^{w-1} (c'_{i}b_{k_0} - \Im_\rho( c'_{i}))z^{k_i}.
	\end{align*}
	Hence, we have
	\begin{equation}\label{eq5}
		c = 1 +  \sum_{i = 1}^{w-1} c_{i}z^{k_i}
	\end{equation}
	where $c'_{i}, ~ d_i\in  \mathbb{F}_q^{*}$ and $a_{i}, k_i\in \mathbb{Z}_{+}$ with $a_{i}, ~k_i< n$ and $c_i= c'_{i}b_{k_0} - \Im_\rho( c'_{i})$, for  $i= 1,\dots,w-1$.

	Define $L_{i}= N_{a_i}^{\rho,\Im_\rho}(\beta_{i})$ and $P_{j}=\sum_{i = 1}^{w-1} c_{i}L_{i}^{j} = c(\beta_i^{j})-1$ and consider a polynomial $p(t)\in \mathbb{F}_q[t]$ given by
	\begin{align*}
		p(t)=& \prod_{i=1}^{w-1} (t-L_i^{m})\\
		=& t^{w-1}+p_1t^{w-2}+\cdots+p_{w-2}t+p_{w-1}\in \mathbb{F}_q[t].
	\end{align*}
	Then
	$L_{i}^m = N_{a_i}^{\rho,\Im_\rho}(\beta_{i}^m) =  N_{a_i}^{\rho,\Im_\rho}(N_{l+mi}^{\rho,\Im_\rho}(\beta^m)) = N_{l+mi}^{\rho,\Im_\rho}(N_{a_i}^{\rho,\Im_\rho}(\beta^m)).$

	If $N_{l+mi}^{\rho,\Im_\rho}(N_{a_i}^{\rho,\Im_\rho}(\beta^m)) = 1$, then
	\begin{align*}
		1 =& N_{l+mi}^{\rho,\Im_\rho}(N_{a_i}^{\rho,\Im_\rho}(\beta^m)) \\
		& = \rho( N_{l+mi-1}^{\rho,\Im_\rho}(N_{a_i}^{\rho,\Im_\rho}(\beta^m)))N_{a_i}^{\rho,\Im_\rho}(\beta^m)+\Im_\rho(N_{l+mi-1}^{\rho,\Im_\rho}(N_{a_i}^{\rho,\Im_\rho}(\beta^m)))\\
		=&\rho(1)N_{a_i}^{\rho,\Im_\rho}(\beta^m).
	\end{align*}
	This implies $N_{a_i}^{\rho,\Im_\rho}(\beta^m) = 1$, which contradicts the assumption.

	Therefore, $L_{i}^m = N_{l+mi}^{\rho,\Im_\rho}(N_{a_i}^{\rho,\Im_\rho}(\beta^m))\neq 1$. Hence, $1$ is not a root of polynomial $p(t)$, i.e., $p(1)\neq 1$.

	Next, we have
	\begin{align*}
		0 =& \sum_{i = 1}^{w-1} c_{i}L_{i}^{l}p(L_{i}^{m})= \sum_{i = 1}^{w-1} c_{i}L_{i}^{l}(L_{i}^{m(w-1)}+ p_{1}L_{i}^{m(w-2)}+\cdots+p_{w-2}L_{i}^{m}+p_{w-1})\\
		=&  \sum_{i = 1}^{w-1} c_{i}(L_{i}^{l+m(w-1)}+ p_{1}L_{i}^{l+m(w-2)}+\cdots+p_{w-2}L_{i}^{l+m}+p_{w-1}L_{i}^{l})\\
		=&  \sum_{i = 1}^{w-1} c_{i}L_{i}^{l+m(w-1)} +p_{1} \sum_{i = 1}^{w-1} c_{i}L_{i}^{l+m(w-2)}+\cdots+p_{w-2}\sum_{i = 1}^{w-1} c_{i}L_{i}^{l+m}\\ &+ p_{w-1}\sum_{i = 1}^{w-1} c_{i}L_{i}^{l}\\
		=&P_{l+m(w-1)}+p_1P_{l+m(w-2)}+\cdots+p_{w-2}P_{l+m}+ p_{w-1}P_l
	\end{align*}
	i.e.,
	\begin{equation}\label{6eqn}
		P_{l+m(w-1)}+p_1P_{l+m(w-2)}+\cdots+p_{w-2}P_{l+m}+ p_{w-1}P_l=0.
	\end{equation}
	Since $c(z)\in\Omega_\textswab{f}(\mathcal{L})= Rg(z)/\mathcal{A}\textswab{f}(z)$, so $c(\beta_{i}^{l+mi})= 0$, for $i= 0,1,\dots,\bigtriangleup-2$. Then $P_{l+mi}=c(\beta_{i}^{l+mi})-1= -1 $. From Equation \ref{6eqn}, we have $-1+p_{1}(-1)+\cdots+p_{w-2}(-1)+p_{w-1}(-1)=0$. This implies $p(1)=0$, which contradicts the assumption. Therefore, $d_{\mathcal{L}}\geq \bigtriangleup$.
\end{proof}
The following theorem is an extension of the Theorem $\ref{6t1}$ and gives a better BCH lower bound for the minimum distance of skew polycyclic codes over $\mathbb{F}_{q}$.
\begin{theorem}\label{6th2}
	Let $\mathcal{L}=\langle g(z) \rangle_{n,q}^{\rho,\Im_\rho}$ be an skew polycyclic code. Suppose there exists  $ \beta\in \bar{\mathbb{F}_q}$ and $l, m_{1},m_{2}\in\mathbb{Z}_{+}\cup \{0\} $ such that $(m_{1},m_{2})\neq (0,0)$, $g(\beta_{i_1}^{l+m_{1}i_{1}+m_{2}i_{2}})=0$, where $\beta_{i_1}= N_{l+m_1i_1+m_2i_2}^{\rho,\Im_\rho}(\beta)$ for $i_1= 0,1,\dots,\bigtriangleup-2$ and $i_2= 0,1,\dots,s$. If $N_{i}(\beta^{m_j})\neq 1$, $N_{l+m_1i_1+m_2i_2-1}(N_{i}(\beta^{m_j}))=1$ and $N_{i}^{\rho,\Im_\rho}(N_{l+m_1i_1+m_2i_2}^{\rho,\Im_\rho}
	(\beta))=N_{l+m_1i_1+m_2i_2}^{\rho,\Im_\rho}(N_{i}^{\rho,\Im_\rho}(\beta))$ $\forall$ $i= 1,2,\dots,n -1$, $j= 1,2$, then $d_{\mathcal{L}}\geq \bigtriangleup +s$.
\end{theorem}
\begin{proof}
	From Theorem \ref{6t1}, we have $d_{\mathcal{L}}\geq \bigtriangleup$. If possible, let there exists a polynomial $c(z)\in\Omega_\textswab{f}(\mathcal{L})$ with $w_H(c)=w$ such that $\bigtriangleup \leq w < \bigtriangleup +s$. Then from equation (\ref{eq5}), $$c(z) = 1+ \sum_{i = 1}^{w-1} c_{i}z^{a_{i}}$$
	where $c_{i}\in \mathbb{F}_q^{*}$ and $a_{i}\in \mathbb{Z}_{+}$ with $a_{i}< n$, for  $i= 1,\dots,w-1$.

	Now, define $L_{i}= N_{a_i}^{\rho,\Im_\rho}(\beta_{i})$ and $P_{j}=\sum_{i = 1}^{w-1} c_{i}L_{i}^{j} = c(\beta_{i_1}^{j})-1$ and consider the polynomials $p(t), q(t)\in \mathbb{F}_q[t]$ defined by
	\begin{align*}
		p(t)=& \prod_{i_1=1}^{\bigtriangleup-2} (t-L_{i_1}^{m_1})\\
		=& t^{\bigtriangleup-2}+p_1t^{\bigtriangleup-3}+\cdots+p_{\bigtriangleup-3}t+p_{\bigtriangleup-2}\in \mathbb{F}_q[t]\\ \text{and}~~
		q(t)=& \prod_{i_2=1}^{w-1} (t-L_{i_2}^{m_2})\\
		=&t^{w-\bigtriangleup+1}+q_1t^{w-\bigtriangleup}+\cdots+q_{w-\bigtriangleup}t+q_{w-\bigtriangleup+1}\in \mathbb{F}_q[t].
	\end{align*}
	Let $r(t)=p(t)q(t)\in \mathbb{F}_q[t]$ and $L_{i_j}^{m_j} = N_{a_{i_j}}^{\rho,\Im_\rho}(\beta_{i_1}^{m_j}) = N_{a_{i_j}}^{\rho,\Im_\rho}(N_{l+m_1i_1+m_2i_2}^{\rho,\Im_\rho}(\beta^{m_j})) =\\ N_{l+m_1i_1+m_2i_2}^{\rho,\Im_\rho}(N_{a_{i_j}}^{\rho,\Im_\rho}(\beta^{m_j}))$. If $N_{l+m_1i_1+m_2i_2}^{\rho,\Im_\rho}(N_{a_{i_j}}^{\rho,\Im_\rho}(\beta^{m_j})) = 1,$ then
	\begin{align*}
		1 =& N_{l+m_1i_1+m_2i_2}^{\rho,\Im_\rho}(N_{a_{i_j}}^{\rho,\Im_\rho}(\beta^{m_j})) \\
		= &\rho( N_{l+m_1i_1+m_2i_2-1}^{\rho,\Im_\rho}(N_{a_{i_j}}^{\rho,\Im_\rho}(\beta^{m_j})))N_{a_{i_j}}^{\rho,\Im_\rho}(\beta^{m_j})+ \Im_\rho(N_{l+m_1i_1+m_2i_2-1}^{\rho,\Im_\rho}(N_{a_{i_j}}^{\rho,\Im_\rho}(\beta^{m_j})))\\
		=&\rho(1)N_{a_{i_j}}^{\rho,\Im_\rho}(\beta^{m_j}).
	\end{align*}
	This implies $N_{a_{i_j}}^{\rho,\Im_\rho}(\beta^{m_j}) = 1$, a contradiction. Therefore, $L_{i_j}^{m_j}\neq 1$, for $j=1,2$. Since $L_{i_1}^{m_1}$ and $L_{i_2}^{m_2}$ are roots of $p(t)$ and $q(t)$, respectively. Hence, $1$ is not a root of $r(t)=p(t)q(t)$. Also,
	\begin{align*}
		0 =& \sum_{i = 1}^{w-1} c_{i}L_{i}^{l}p(L_{i}^{m_1})q(L_{i}^{m_2})\\
		=& \sum_{i = 1}^{w-1} c_{i}L_{i}^{l}(L_{i}^{m_1(\bigtriangleup-2)}+ p_{1}L_{i}^{m_1(\bigtriangleup-3)}+\cdots+p_{\bigtriangleup-2})(L_{i}^{m_2(w-\bigtriangleup-1)}\\
		&+ q_{1}L_{i}^{m_2(w-\bigtriangleup)}+\cdots+q_{w-\bigtriangleup+1})\\
		=&  \sum_{i = 1}^{w-1} c_{i}[(L_{i}^{l+m_1(\bigtriangleup-2)}+ p_{1}L_{i}^{l+m_1(\bigtriangleup-3)}+\cdots+p_{\bigtriangleup-2}L_{i}^{l})(L_{i}^{m_2(w-\bigtriangleup-1)}\\
		&+ q_{1}L_{i}^{m_2(w-\bigtriangleup)}+\cdots+q_{w-\bigtriangleup+1})]\\
		=&  \sum_{i = 1}^{w-1} c_{i}[(L_{i}^{l+m_1(\bigtriangleup-2)+m_2(w-\bigtriangleup-1)}+ p_{1}L_{i}^{l+m_1(\bigtriangleup-3)+m_2(w-\bigtriangleup-1)}+\cdots\\
		&+p_{\bigtriangleup-2}L_{i}^{l+m_2(w-\bigtriangleup+1)})+\cdots+q_{w-\bigtriangleup+1}(L_{i}^{l+m_1(\bigtriangleup-2)}+ p_{1}L_{i}^{l+m_1(\bigtriangleup-3)}\\
		&+\cdots+p_{\bigtriangleup-2}L_{i}^{l})]\\
		=&P_{l+m_1(\bigtriangleup-2)+m_2(w-\bigtriangleup+1)}+p_1P_{l+m_1(\bigtriangleup-3)+m_2(w-\bigtriangleup+1)}+\cdots+ \\ & p_{\bigtriangleup -2}P_{l+m_2(w-\bigtriangleup+1)}+\cdots+q_1(P_{l+m_1(\bigtriangleup-2)+m_2(w-\bigtriangleup)}+\\& p_1P_{l+m_1(\bigtriangleup-3)+m_2(w-\bigtriangleup)}+\cdots+ p_{\bigtriangleup -2}P_{l+m_2(w-\bigtriangleup)})
		+\cdots+q_{w-\bigtriangleup+1}\\&     (P_{l+m_1(\bigtriangleup-2)}+p_1P_{l+m_1(\bigtriangleup-3)}+\cdots+  p_{\bigtriangleup -2}P_{l}).
	\end{align*}
	i.e.,
	\begin{eqnarray}
		&&\nonumber P_{l+m_1(\bigtriangleup-2)+m_2(w-\bigtriangleup+1)}+p_1P_{l+m_1(\bigtriangleup-3)+m_2(w-\bigtriangleup+1)}+\cdots\\
		&&\nonumber \hspace{.7cm}
		+ p_{\bigtriangleup -2}
		P_{l+m_2(w-\bigtriangleup+1)}+
		\cdots+q_1(P_{l+m_1(\bigtriangleup-2)+m_2(w-\bigtriangleup)}+\\
		&&\nonumber  \hspace{1cm}p_1P_{l+m_1(\bigtriangleup-3)+m_2(w-\bigtriangleup)}+\cdots+ p_{\bigtriangleup -2}P_{l+m_2(w-\bigtriangleup)})+ \cdots\\
		&& \hspace{1.3cm}+q_{w-\bigtriangleup+1}(P_{l+m_1(\bigtriangleup-2)}
		+p_1P_{l+m_1(\bigtriangleup-3)}+\cdots+ p_{\bigtriangleup -2}P_{l}) =0.
	\end{eqnarray}
	Since $c(z)\in\Omega_\textswab{f}(\mathcal{L})$ and $c(\beta_{i_1}^{l+m_1i_1+m_2i_2})=0$, $P_{l+m_1i_1+m_2i_2}=c(\beta_{i_1}^{l+m_1i_1+m_2i_2})-1= -1 $, for $i_1= 0,\dots,\bigtriangleup -2$, $i_2= 0,\dots,s$ and  $\bigtriangleup \leq w < \bigtriangleup +s$. Hence, from the above equations, we have
	$(-1-p_1+\cdots+p_{\bigtriangleup-3}-p_{\bigtriangleup-2})(1+q_1+\cdots+q_{w-\bigtriangleup}-q_{w-\bigtriangleup+1})=0$. This implies $r(1)=0$, which contradicts the assumption. Thus, $w\geq \bigtriangleup +s$, i.e., $d_{\mathcal{L}}\geq \bigtriangleup +s$.
\end{proof}
In view of Theorem \ref{6t1},  we provide an extension of Theorem \ref{6th2}, which can be proved inductively.
\begin{theorem}\label{6th3}
	Let $\mathcal{L}=\langle g(z) \rangle_{n,q}^{\rho,\Im_\rho}$ be an skew polycyclic code. Suppose that there exist $ \beta\in \bar{ \mathbb{F}_q}$ and $l,m_{1},m_{2},\dots m_r\in\mathbb{Z}_{+}\cup\{0\} $ such that $(m_{1},m_{2}\dots m_r)\neq (0,0)$, $g(\beta_{i_1}^{l+m_{1}i_{1}+\cdots+m_{r}i_{r}})=0$, where $\beta_{i_1}= N_{l+m_1i_1+\cdots+m_{r}i_{r}}^{\rho,\Im_\rho}(\beta)$ for $i_1= 0,\dots,\bigtriangleup-2$ and  $i_k= 0,\dots,s_k$ and $k=2,3,\dots, r$. If $N_{i}(\beta^{m_j})\neq 1$, $N_{l+m_1i_1+m_2i_2+\cdots+m_{r}i_{r}-1}(N_{i}(\beta^{m_j}))=1$ and $N_{i}^{\rho,\Im_\rho}(N_{l+m_1i_1+\cdots+m_{r}i_{r}}^{\rho,\Im_\rho}(\beta))=N_{l+m_1i_1+\cdots+m_{r}i_{r}}^{\rho,\Im_\rho}(N_{i}^{\rho,\Im_\rho}(\beta))$ for all $i= 1,\dots,n -1$, $j= 1,2,\dots, r$, then $d_{\mathcal{L}}\geq \bigtriangleup + \sum_{k=2}^{r}s_k$.
\end{theorem}
\begin{corollary}
	Let $\mathcal{L}=\langle g(z) \rangle_{n,q}^{\rho,\Im_\rho}$ with $q \geq n+1$. If $\beta \in \bar{\mathbb{F}_q}$ and $l,c_{1},c_{2},\dots, c_r\in\mathbb{Z}_{+}\cup\{0\} $ such that  $(c_{1},c_{2},\dots, c_r)\neq (0,0,\dots,0)$,
	$g(z)= lclm \{z-\beta^{l+\sum_{k=1}^{r}i_kc_k} : i_1=0,\dots,\bigtriangleup,~ i_k=0,\dots, s_k,~ \bigtriangleup + \sum_{k=2}^{r}s_k=n -k+1\},$
	$N_{i}(\beta^{c_j})\neq 1$, $N_{l+c_1i_1+c_2i_2+\cdots+c_{r}i_{r}}(N_{i}(\beta^{c_j}))=1$ and $N_{i}^{\rho,\Im_\rho}(N_{l+c_1i_1+\cdots+c_{r}i_{r}}^{\rho,\Im_\rho}(\beta)) =N_{l+c_1i_1+\cdots+c_{r}i_{r}}^{\rho,\Im_\rho}(N_{i}^{\rho,\Im_\rho}(\beta))$ for all $i= 1,\dots,n -1$, $j= 1,2,\dots, r$, then $\mathcal{L}$ is an MDS code.
\end{corollary}
\begin{proof}
	Since $deg (g(z))\leq n -k$, and the singleton bound is given by $d_{\mathcal{L}}\leq  n-dim(\mathcal{L}) +1= n-(n-deg(g(z)))+1 \leq n -k+1$. From Theorem \ref{6th3}, it follows that $d_{\mathcal{L}}\geq n -k+1$. This implies $d_{\mathcal{L}}=n -k+1.$
\end{proof}
\section{Comparative Summary}
In Table \ref{tabMDS}, we have provided the MDS codes in $\mathbb{F}^n_q$ for $q=8,~11$. It can be easily seen that for some code parameters, cyclic and constacyclic codes do not exist, but one can construct polycyclic codes with the same code parameters. In Table \ref{tab1}, we have also shown that for some code parameters, polycyclic codes do not exist, but the corresponding skew polycyclic codes exist. Here, we have performed all the calculations using the Magma software \cite{Bosma}.
\begin{longtable}[h]{|l|l|l|l|l|l|}
	\hline	
	$q$  & $n=$  & $k$ & $d$ & $g$ such that   & $a$ such that \\
	& $\deg~ \textswab{f}$&  &   &  $g$ divides $\textswab{f}$  &  $\textswab{f}= z^n-a$\\
	\hline	
	$8$ & $2$ & $1$ & $2$ & $z+w^6, z+w^2, z+w^5, z+w,z+w^4,z+1$,& $1$,$w$,$ w^2$,$w^6$\\
	&     &   &   & $z+w^3$ & $w^3,w^4,w^5$\\
	\hline
	$8$& $3$& $2$ & $2$ &   $z+w^6, z+w^2, z+w^5, z+w,z+w^4,z+1,$& $1$,$w$, $w^2$,$w^6$\\
	&     &   &   & $z+w^3$ & $w^3,w^4,w^5$\\
	\hline
	& & $1$ & $3$ &  $z^2+z+w,z^2+w^3z+1, z^2+wz+w^3 $& $\nexists$\\
	& & &  & $ z^2+w^4z+w^2, z^2+w^2z+w^5, z^2+w^2z+w,$&\\
	&     &   &   & $z^2+w^5z+w^4,z^2+w^5z+w,z^2+z+w^4$&\\
	&     &   &   & $z^2+z+w^2,z^2+w^3z+w^3,z^2+w^6z+w^6$&\\
	&     &   &   & $z^2+w^3z+w,z^2+w^6z+w^2,z^2+wz+w^6$&\\
     \hline
	&     &   &   & $z^2+w^6+1,z^2+wz+w^4,z^2+w^4z+w^5$&\\	
	&     &   &   & $z^2+w^4z+w^3,z^2+w^2z+w^6,z^2+w^5z+w^5$&\\	
	\hline
	$8$ & $4$ & $3$ & $2$ & $z+w^6, z+w^2, z+w^5, z+w,z+w^4,z+1,$& $1,w, w^2,w^3$\\
	&     &   &   & $z+w^3$ & $w^4,w^5,w^6$\\
	\hline	
	&$4$& & & for codes of &	\\
	&     &   & & the type $[4,1,4]$ and   $[4,2,3]$ $\nexists ~ \textswab{f}$& \\
	\hline
	$8$ & $5$ & $4$ & $2$ & $z+w^6, z+w^2, z+w^5, z+w,z+w^4,z+1,$& $1,w, w^2,$\\
	&     &   &   & $z+w^3$ & $w^3,w^4,w^5,$\\
	&     &   &   &  & $w^6$\\
	\hline
	&$5$ & & & for codes of &	\\
	&  &   &  & the type $[5,3,3]$, $[5,1,5]$ and   $[5,2,4]$ $\nexists ~ \textswab{f}$& \\
	\hline
	$8$ & $6$ & $5$ & $2$ & $z+w^6, z+w^2, z+w^5, z+w,z+w^4,z+1,$& $1,w, w^2,w^6$\\
	&     &   &   & $z+w^3$ & $w^3,w^4,w^5,$\\
	\hline
	$8$&$6$ & & & for codes of &	\\	
	&  &   &  & the type $[6,1,6]$, $[6,4,3]$, $[6,2,5]$ and   $[6,3,4]$ $\nexists ~ \textswab{f}$& \\	
	\hline
	&$7$ & & & for codes of &	\\	
	&  &   &  & the type $[7,2,6]$, $[7,1,7]$, $[7,5,3]$,&\\
	& & & & $[7,6,2]$, $[7,3,5]$ and   $[7,4,4]$ always $\exists ~ \textswab{f}$ & $1$\\
	\hline
	$11$ & $2$& $1$& $2$ & $z+10$& $1$\\	
	\hline
	$11$& $3$& $2$& $2$ & $z+10$& $1$\\	
	\hline
	
	& & $1$& $3$ & $z^2+9z+1$& $\nexists$\\	
	\hline
	$11$ & $4$& $3$ & $2$ & $z+10$& $1$\\
	\hline
	& & $2$& $3$ & $z^2+9z+1$& $\nexists$\\
	\hline
	$11$ & $5$& $4$ & $2$ & $z+10$& $1$\\
	\hline
	
	& & $2$& $4$ & $z^3+8z^2+3z+10$& $\nexists$\\
	\hline	
	& & $3$& $3$ & $z^2+9z+1$& $\nexists$\\
	\hline	
	
	$11$ & $6$& $5$ & $2$ & $z+10$& $1$\\
	\hline
	& & $2$& $5$ & $z^4+7z^3+6z^2+7z+1$& $\nexists$\\
	\hline
	& & $3$& $4$ & $z^3+8z^2+3z+10$& $\nexists$\\
	\hline
	& & $4$& $3$ & $z^2+9z+1$& $\nexists$\\
	\hline
	$11$ & $7$& $6$ & $2$ & $z+10$& $1$\\
	\hline
	& & $2$& $6$ & $z^5+6z^4+10z^3+z^2+5z+10$& $\nexists$\\
	
	\hline
	& & $3$& $5$ & $z^4+7z^3+6z^2+7z+1$& $\nexists$\\
	\hline
	& & $4$& $4$ & $z^3+8z^2+3z+10$& $\nexists$\\
	\hline	
	& & $5$& $3$ & $z^2+9z+1$& $\nexists$\\
	\hline
	$11$ & $8$& $7$ & $2$ & $z+10$& $1$\\
	\hline
	
	& & $3$& $6$ & $z^5+6z^4+10z^3+z^2+5z+10$& $\nexists$\\
	\hline
	& & $4$& $5$ & $z^4+7z^3+6z^2+7z+1$& $\nexists$\\
	\hline	
	& & $5$& $4$ & $z^3+8z^2+3z+10$& $\nexists$\\
	\hline		
	$11$& $8$ & $2$& $7$ & $z^6+5z^5+4z^4+2z^3+4z^2+5z+1$& $\nexists$\\
	\hline		
	& & $6$& $3$ & $z^2+9z+1$& $\nexists$\\
	\hline		
	$11$ & $9$& $8$ & $2$ & $z+10$& $1$\\
	\hline
	& & $4$& $6$ & $z^5+6z^4+10z^3+z^2+5z+10$& $\nexists$\\
	\hline
	& & $5$& $5$ & $z^4+7z^3+6z^2+7z+1$& $\nexists$\\
	\hline		
	& & $6$& $4$ & $z^3+8z^2+3z+10$& $\nexists$\\
	\hline		
	& & $3$& $7$ & $z^6+5z^5+4z^4+2z^3+4z^2+5z+1$& $\nexists$\\
	\hline		
	& & $7$& $3$ & $z^2+9z+1$& $\nexists$\\
	\hline		
	& & $2$& $8$ & $z^7+4z^6+10z^5+9z^4+2z^3+z^2+7z+10$& $\nexists$\\
	\hline
	$11$ & $10$& $9$ & $2$ & $z+10$& $1$\\
	\hline
	& & $5$& $6$ & $z^5+6z^4+10z^3+z^2+5z+10$& $\nexists$\\
	\hline
	& & $6$& $5$ & $z^4+7z^3+6z^2+7z+1$& $\nexists$\\
	\hline		
	& & $7$& $4$ & $z^3+8z^2+3z+10$& $\nexists$\\
	\hline		
	& & $4$& $7$ & $z^6+5z^5+4z^4+2z^3+4z^2+5z+1$& $\nexists$\\
	\hline		
	& & $8$& $3$ & $z^2+9z+1$& $\nexists$\\
	\hline		
	& & $3$& $8$ & $z^7+4z^6+10z^5+9z^4+2z^3+z^2+7z+10$& $\nexists$\\
	\hline	
	& & $2$& $9$ & $z^8+3z^7+6z^6+10z^5+4z^4+10z^3+6z^2+3z+1$& $\nexists$\\
	\hline
	\caption{ MDS codes in $\mathbb{F}^n_q$ for $q=8,~11$}\label{tabMDS}	
\end{longtable}
\subsection{Comparison of  polycyclic and skew polycyclic codes}
It is seen that skew polycyclic codes have better parameters over the polycyclic codes. To make it clear, we consider
$$\textswab{f}(z)= z^5+z^4+z^3+wz^2+1 \in \mathbb{F}_{9}[z;\rho].$$  Then we obtain $16$ MDS skew polycyclic codes with parameters $[5,2,4]_9$ and $[5,3,3]_9$. The generator polynomials for the codes having parameters $[5,2,4]_9$ are

$wz^3 + 2z^2 + w^2z + 1,~w^2z^3 + w^5z^2 + w^3z + w,~w^3z^3 + w^6z^2 + 2z + w^2,~2z^3 + w^7z^2 + w^5z + w^3
,~w^5z^3 + z^2 + w^6z + 2,~w^6z^3 + wz^2 + w^7z + w^5,~w^7z^3 + w^2z^2 + z + w^6,~z^3 + w^3z^2 + wz + w^7;$

and codes with parameters $[5,3,3]_9$ are
$w^7z^2 + w^2z + 1,~z^2 + w^3z + w,wz^2 + 2z + w^2,~w^2z^2 + w^5z + w^3,~w^3z^2 + w^6z + 2,~2z^2 + w^7z + w^5,~w^5z^2 + z + w^6,w^6z^2 + wz + w^7$.

On the other hand, we get only one MDS polycyclic code with parameters $[5,3,3]_{9}$ and its generator polynomial $z^2+w^6z+w^7$. There does not exist any polycyclic code with parameters $[5,2,4]_9$.
\begin{example}
	Consider $\textswab{f}(z)= z^5+wz^4+z^3+wz^2+1 \in \mathbb{F}_8[z;\rho]$. Here, we obtain $21$ MDS skew polycyclic codes with parameters $[5,4,2]_8$, $[5,1,5]_8$, and $[5,3,3]_8$. If $\rho=id$, then we get only $2$ MDS polycyclic codes with parameters $[5,4,2]_8$ and $[5,3,3]_8$ having corresponding generator polynomials $z+w^6$ and $z^2+w^3x+w^4$, respectively.
\end{example}

\begin{example}
	Consider $\textswab{f}(z)= z^5+z^4+z^3+wz^2+1 \in \mathbb{F}_{9}[z;\rho]$.
	Here, we obtain $16$ MDS skew polycyclic codes with parameters $[5,2,4]_9$ and $[5,3,3]_9$. If $\rho=id$, then we get only one MDS polycyclic code with parameters $[5,3,3]_{9}$ and generator polynomial $z^2+w^6z+w^7$.
\end{example}
\begin{example}
	Consider $\textswab{f}(z)=wz^5+z^4+z^3+z^2+1 \in \mathbb{F}_{8}[z;\tau]$. Here, we obtain $21$ MDSskew polycyclic codes with parameters $[5,2,4]_8$, $[5,4,2]_8$  and $[5,1,5]_8$. If $\tau=id$, then we get only $2$ MDS polycyclic codes with parameters $[5,4,2]_{8}$ and $[5,1,5]_{8}$ having corresponding generator polynomials $z + 4$ and $z^4 + 2z^3 + 3z^2 + 4z + 4,$ respectively.
\end{example}
\section{Conclusion}
In this paper, we have considered skew generalized polycyclic codes over the iterated skew polynomial ring $ \mathscr{R}[x_1;\tau_1,\delta_{\tau_1}][x_2;\tau_2,\delta_{\tau_2}]$, that are invariant under pseudo-linear transformation of $\mathscr{R}^n$ where $\mathscr{R}$ is a ring with unity or a finite field $\mathbb{F}_q$. Further, for non-zero derivations, we have studied the BCH lower bounds for the minimum distance of these codes. As an application of our established results, we have also included some examples of MDS codes. Finally, we have provided a comparative summary of our proposed work with polycyclic codes, constacyclic codes and cyclic codes, which is shown in Table \ref{tab1} and Table \ref{tabMDS}.

\begin{landscape}
\begin{table}
\renewcommand{\arraystretch}{.8}
			\begin{center}
				
				\begin{tabular}{|c|c|c|c|c|c|c|c|c|c|c|c|c|c|c|}
					
					\hline	
					$q$  & $n$  & $\textswab{f}(z)$ & $d$ & Parameters   & No. of   &Parameters& No. of  \\
					&&&&of polycyclic& MDS&of skew polycyclic codes & MDS\\
					&&&&codes&polycyclic &&skew polycyclic codes \\
					&&&&& & &\\
					\hline
					$8$& $8$ &  $ wz^8+w^2z^7+wz^6+wz^4 $ & $3$ &  & $\nexists$& $[8,6,3]_8$&$7$ \\
					&&$+z^2+1$&&&&&\\
					$8$& $9$ & $ z^9 + z^7+wz^6+wz^5 + w^2z^4  $ & $2,3$ &$[9,8,2]_8$  & $1$& $[9,8,2]_8$&$14$ \\
					&&$+wz^3+ w^2z ^2 + w$&&& & $[9,7,3]_8$&\\
					$8$& $9$ & $ z^9 + wz^5 + z^4 + z^2 + 1 $ & $2,3$ &$[9,8,2]_8$  & $1$& $[9,8,2]_8$&$21$ \\
					&&&&& & $[9,7,3]_8$&\\
					
					$9$& $8$ & $ wz^8+w^2z^7+wz^6+wz^4 $ & $2,3$ &  & $\nexists$&$[8,7,2]_9$ &$23$ \\
					&&$+z^2+1$&&& &$[8,6,3]_9$& \\
					$9$& $5$ & $ wz^5+z^4+z^3+z^2+1 $ & $2,3,4,5$ &$[5,4,2]_9$  & $1$&$[5,4,2]_9$ &$24$ \\
					&&&&& &$[5,3,3]_9$& \\
					&&&&& &$[5,2,4]_9$& \\	
					&&&&& &$[5,1,5]_9$& \\	
					$9$& $6$ & $ wz^6+wz^4+z^3+wz^2+1 $ & $2,3,4$ &$[6,5,2]_9$  & $2$&$[6,5,2]_9$ &$40$ \\
					&&&&& &$[6,4,3]_9$& \\
					&&&&& &$[6,3,4]_9$& \\
					\hline
				\end{tabular}\caption{Comparison between the number of MDS skew polycyclic codes and polycyclic codes}
				\label{tab1}
			\end{center}
		\end{table}
	\end{landscape}

\section*{Acknowledgements}
The authors are thankful to the Department of Science and Technology (DST), Govt. of India for financial support under Ref No.- DST/INSPIRE/03/2016/ 001445 and Indian Institute of Technology Patna for providing the research facilities.
\section*{Data Availability}
The authors declare that [the/all other] data supporting the findings of this study are available within the article. Any clarification may be requested from the corresponding author, provided it is essential. \\
\textbf{Competing interests}: The authors declare that there is no conflict of interest regarding the publication of this manuscript.


\begin{thebibliography}{20}
 \bibitem{Adel}  Alahmadi, A.,  Dougherty, S.,  Leroy, A.,  Sol{\`e}, P.:  On the duality and the direction of polycyclic codes. Adv. Math. Commun. $\boldsymbol{10}(4)$, $921-929$ $(2016)$.
	\bibitem{Blake1}  Blake, I. F.: Codes over certain rings. Inform. and Control, $\boldsymbol{20}(4)$, $396-404$ $(1972)$.
	\bibitem{Blake} Blake, I. F.: Codes over integer residue rings. Inform. Control, $\boldsymbol{29}$, $295-300$  $(1975)$.
	\bibitem{Bosma}  Bosma, W., Cannon, J.:  Handbook of Magma Functions, Univ. of Sydney, Sydney $(1995)$.
	\bibitem{Boucher1}  Boucher, D., Geiselmann, W.,   Ulmer, F.: Skew cyclic codes. Appl. Algebra Engrg. Comm. Comput. $\boldsymbol{18}(4)$ , $379-389$ $(2007)$.
	\bibitem{MH} Boulagouaz, M., H., Leroy, A.: $(\sigma, \delta)$-codes, Adv. Math. Commun. $\boldsymbol{7}(4)$, $463-474$ $(2013)$.
	\bibitem{Boucher} Boucher, D.,  Ulmer, F.: Linear codes using skew polynomials with automorphisms and derivations. Des. Codes Cryptogr. $\boldsymbol{70}(3)$, $405-431$ $(2014)$.
	\bibitem{Boucher2}   Boucher, D.,  Ulmer, F.: Coding with skew polynomial rings. J. Symbolic Comput. $\boldsymbol{44}(12)$, $1644-1656$  $(2009)$.
    \bibitem{Garani} Roy, S., Garani, S.S.: Two-dimensional algebraic codes for multiple burst error
correction. IEEE Commun. Lett. 23(10), 1684–1687 (2019)
\bibitem{Yoon} Yoon, S.W., Moon, J.: Two-dimensional error-pattern-correcting codes. IEEE
Trans. Comm. 63(8), 2725–2740 (2015)
	\bibitem{Grassl}  Grassl, M.: Code Tables: Bounds on the parameters of various types of codes available at \emph{http://www.codetables.de/}	accessed on 20/03/2020.
	\bibitem{Hartmann}  Hartmann, C. R., Tzeng, K. K.: Generalizations of the BCH bound, Inf. Control. $\boldsymbol{20}$, $489-498$  $(1972)$.
	\bibitem{Imai} Imai, H.:  A theory of two-dimensional cyclic codes. Inf. Control, $\boldsymbol{34}$, $1-21$  $(1977)$.
    \bibitem{Haji} Hajiaghajanpour, N., Khashyarmanesh, K.: Two dimensional double cyclic
codes over finite fields. Applicable Algebra in Engineering, Communication and
Computing, 1–25 (2023)
	\bibitem{N} Jacobson, N.: Pseudo-linear transformations, Ann. Math., $\boldsymbol{38}(2)$, $484-507$  $(1937)$.
	\bibitem{Ikai}  Ikai, T., Kosako, H., Kojima, Y.:  Two dimensional cyclic codes, Electron Comm. Japan, $\boldsymbol{57}(4)$, $27-35$ $(1975)$.
	
	\bibitem{Lopez1} L\'opez-Permouth,  S. R., Parra-Avila, B. R., Szabo, S.:  Dual generalizations of the concept of cyclicity of codes. Adv. Math. Commun. $\boldsymbol{3}(3)$,  $227-234$ $(2009)$.
			
			\bibitem{Lopez}  L\'opez-Permouth, S. R.,  Sergio, R., \"Ozadam, H., \"Ozbudak, F., Szabo, S.: Polycyclic codes over Galois rings with applications to repeated-root constacyclic codes.Finite Fields Appl. $\boldsymbol{19}$, $16-38$ $(2013)$.
    \bibitem{Moro}  Martínez-Moro, E.,  Fotue, A.,  Blackford, T.:  On polycyclic codes over a finite chain ring.  Advances in Mathematics of Communications, $\boldsymbol{3}$, $2020$, https://doi.org /10.3934 /amc.2020028.
    \bibitem{Patel}  Patel, S., Prakash, O.: Repeated-root bidimensional $(\mu,\nu)$-constacyclic codes of
length $4p^t.2^r$. International Journal of Information and Coding Theory 5(3-4),
266–289 (2020)
    \bibitem{Prakash}  Prakash, O., Patel, S.: A note on two-dimensional cyclic and constacyclic codes.
J. Algebra Comb. Discrete Struct. Appl., 161–174 (2022)
\bibitem{Rajabi}  Rajabi, Z., Khashyarmanesh, K.: Repeated-root two-dimensional constacyclic
codes of length $2p^s.2^k$. Finite Fields Appl. 50, 122–137 (2018)
\bibitem{Sagar} Sagar, V., Yadav, A., Sarma, R.: Constacyclic codes over $Z_2[ u ]/⟨u^2⟩ \times Z_2[ u ]/⟨u^3⟩$ and the MacWilliams identities. Applicable Algebra in Engineering, Communication and Computing, (2024), DOI: 10.1007/s00200-024-00662-6.
	\bibitem{Z} Sepasdar, Z.,  Khashyarmanesh, K.: Characterizations of some two-dimensional cyclic codes correspond to the ideals of $\mathbb{F}[x, y]/\langle x^s-1,y^{2^k}-1\rangle$. Finite Fields Appl., $\boldsymbol{41}$, $97-112$ $(2016)$.
	\bibitem{A} Sharma, A., Bhaintwal, M.: A class of  $2D$ Skew-cyclic codes over $\mathbb{F}_q+u\mathbb{F}_q$, Appl. Algebra Engrg. Comm. Comput., $\boldsymbol{30}(6)$, $71-90$  $(2019)$.
	\bibitem{Shi}  Shi, M.,  Li, X., Sepasdar, Z.,  Sol{\`e}, P.:   Polycyclic codes as invariant subspaces, Finite Fields Appl., $\boldsymbol{68}$, $101760$  $(2020)$.
	\bibitem{Luis} Tapia Cuiti${\tilde n}$o, L. F.,  Tironi, A. L.: Some properties of skew codes over finite fields, Des. Codes Cryptogr., $\boldsymbol{85}$, $359-380$  $(2017)$.
   \bibitem{Tharkal}  Thakral, R., Dutt, S., Sehmi, R.: Linear complementary pairs of constacyclic
n-D codes over a finite commutative ring. Applicable Algebra in Engineering,
Communication and Computing, 1–14 (2024)
	\bibitem{Gr} Voskoglou, M. G.: Derivations and iterated skew polynomial rings, Int. J. Appl. Math. Informat., $\boldsymbol{5}(2)$, $82-90$ $(2011)$.	
	\bibitem{Li}  Xiuli, L.,  Hongyan, L.: $2$-D skew-cyclic codes over $F_q[x,y;\rho, \theta]$, Finite Fields Appl., $\boldsymbol{25}$, $49-63$  $(2014)$.
	\bibitem{Yildiz} Yildiz, B., Aydin, N., On cyclic codes over $\mathbb{Z}_4+u\mathbb{Z}_4$ and their $\mathbb{Z}_4$-images, Int. J. Inf. Coding Theory,  $\boldsymbol{2}(4)$, $226-237$  $(2014)$.
	\bibitem{Zheng} Zheng, X.,  Kong, B.: Cyclic codes and $\lambda_1 + \lambda_2u + \lambda_3v + \lambda_4uv$-constacyclic codes over $\mathbb{F}_p+ u\mathbb{F}_p + v\mathbb{F}_p + uv\mathbb{F}_p$, Appl. Math. Comput., $\boldsymbol{306}$, $86-91$ (2017).
	
	
	
\end{thebibliography}
	\end{document}